\documentclass[11pt,letterpaper]{article}

\usepackage{amsmath,amsthm,amsfonts,amssymb}
\usepackage{thm-restate}
\usepackage{fullpage}
\usepackage[utf8]{inputenc}
\usepackage[dvipsnames]{xcolor}
\usepackage{xspace,enumerate}
\usepackage[shortlabels]{enumitem}
\usepackage[hypertexnames=false,colorlinks=true,urlcolor=Blue,citecolor=Green,linkcolor=BrickRed]{hyperref}
\usepackage[capitalise]{cleveref}
\usepackage[OT4]{fontenc}
\usepackage[ruled]{algorithm}
\usepackage[noend]{algorithmic}
\usepackage{ifpdf}
\usepackage{todonotes}
\usepackage{enumerate}
\usepackage{authblk}
\usepackage{cite}
\usepackage{thmtools}

\usepackage[margin=1in]{geometry}

\def\polylog{\operatorname{polylog}}

\title{Improved Strongly Polynomial Algorithms for Deterministic MDPs,\\2VPI Feasibility, and Discounted All-Pairs Shortest Paths}
\date{}
\author[1]{Adam Karczmarz\thanks{\texttt{a.karczmarz@mimuw.edu.pl}. Supported by the ERC Consolidator Grant 772346 TUgbOAT.}}
\affil[1]{Institute of Informatics, University of Warsaw and IDEAS NCBR, Poland}

\DeclareMathOperator*{\argmax}{argmax}

\newcommand{\Ot}{\ensuremath{\widetilde{O}}}
\newcommand{\eps}{\ensuremath{\epsilon}}
\newcommand{\dist}{\delta}

\theoremstyle{plain}
\newtheorem{theorem}{Theorem}[section]
\newtheorem{lemma}[theorem]{Lemma}

\newtheorem{observation}[theorem]{Observation}
\newtheorem{definition}[theorem]{Definition}

\newtheorem{remark}[theorem]{Remark}

\newtheorem{problem}[theorem]{Problem}

\begin{document}

\maketitle

\begin{abstract}
\small
  We revisit the problem of finding optimal strategies for deterministic
  Markov Decision Processes (DMDPs), and a closely related problem of
  testing feasibility of systems of $m$ linear
  inequalities on $n$ real variables with at most two variables per inequality (2VPI).
  
  We give a randomized trade-off algorithm solving both problems and running
  in \linebreak $\Ot(nmh+(n/h)^3)$ time using $\Ot(n^2/h+m)$ space for any
  parameter $h\in [1,n]$. In particular, using subquadratic space we get
  $\Ot(nm+n^{3/2}m^{3/4})$ running time, which improves by a polynomial factor upon all
  the known upper bounds for non-dense instances with $m=O(n^{2-\eps})$.
  Moreover, using linear space we match the randomized $\Ot(nm+n^3)$ time bound
  of Cohen and Megiddo [SICOMP'94] that required $\tilde{\Theta}(n^2+m)$ space.

  Additionally, we show a new algorithm for the Discounted All-Pairs Shortest Paths problem,
  introduced by Madani~et~al. [TALG'10], that extends the DMDPs
  with optional end vertices.
  For the case of uniform discount factors, we give a deterministic algorithm running
  in $\Ot(n^{3/2}m^{3/4})$ time, which improves significantly upon
  the randomized bound $\Ot(n^2\sqrt{m})$ of Madani~et~al.

\end{abstract}

\setcounter{page}{1}

\newcommand{\xmax}[1]{{x_{#1}^{\max}}}
\newcommand{\xmin}[1]{{x_{#1}^{\min}}}

\section{Introduction}\label{s:intro}

Linear programming is one of the most fundamental problems in optimization and algorithm design.
Polynomial-time algorithms solving linear programming relying on linear algebraic techniques have been known
for decades~\cite{Karmarkar84, khachiyan1979}, and the recently developed solutions run in the
current matrix multiplication time~\cite{Brand20, CohenLS21, JiangSWZ21}.
All the known polynomial algorithms are \emph{weakly polynomial}. This means
the number of arithmetic operations they perform depends on the magnitude of
the numbers used to describe the $m$ constraints and $n$ variables of a linear program (LP).
It remains a challenging open problem whether there exists a \emph{strongly polynomial}~\cite{Megiddo83b} algorithm
whose worst-case number of arithmetic operations can be bounded by a polynomial in $n$ and $m$ exclusively.
The quest for such an algorithm led to strongly polynomial algorithms
for some important restricted classes of LPs, such as shortest paths,
minimum cost flows~\cite{Orlin93, Tardos85}, ``combinatorial'' LPs~\cite{Tardos86, DadushNV20}, generalized maximum flows~\cite{OlverV20, Vegh17},
and (discounted) Markov Decision Processes~\cite{HansenMZ13, Ye11a}.

One of the simplest classes of LPs for which  no strongly polynomial
algorithm is known is the class of LPs such that each constraint involves at most
two variables -- so-called \emph{two variables per inequality (2VPI)} programs.
Optimizing 2VPI programs can be reduced to solving the uncapacitated generalized minimum cost flow problem,
for which a \emph{combinatorial} (i.e., not relying on algebraic techniques)
weakly polynomial algorithm is known~\cite{Wayne02}.\footnote{Wayne's algorithm~\cite{Wayne02} in fact works
even for the capacitated generalized minimum cost flow problem.}
Interestingly, strongly polynomial algorithms are known
for a seemingly only slightly easier problem of testing the \emph{feasibility} of a 2VPI program,
where there is no objective function to optimize
and any solution satisfying the constraints is sought.

In this paper, we revisit the \emph{2VPI feasibility} problem and a closely related problem
of (discounted) \emph{Deterministic Markov Decision Processes} (DMDPs) in the strongly polynomial regime.

\paragraph{2VPI feasibility.}
The 2VPI feasibility problem has a rich history, and first weakly polynomial time
algorithms~\cite{aspvall1981efficient, AspvallS80} have been developed around the same time as
Khachiyan's breakthrough weakly polynomial algorithm for solving general LPs~\cite{khachiyan1979}. Megiddo~\cite{Megiddo83b} initiated the
quest for strongly polynomial~linear programming algorithms and obtained an $\Ot(mn^3)$ time strongly
polynomial algorithm for 2VPI feasibility using his \emph{parametric search} technique~\cite{Megiddo83a}.
Subsequently, Cohen and Megiddo~\cite{CohenM94} found faster strongly polynomial algorithms:
a deterministic one running in $O(mn^2(\log^2{n}+\log{m}))$, and a randomized one running
in $\Ot(n^3+mn)$ time. Hochbaum and Naor~\cite{HochbaumN94} showed a slightly faster
deterministic algorithm running in $O(mn^2\log{m})$ time using a different approach.
These algorithms~\cite{CohenM94, HochbaumN94} remain the current state of the art
and all rely on parametric search in some way.
However, their drawback is that they use $\Theta(n^2+m)$ space
which seems inherent (as will be explained later on).
This stands in sharp contrast with the earlier algorithm of Megiddo~\cite{Megiddo83b} which needs linear space.

Recently, Dadush~et~al.~\cite{dadush2021accelerated} showed a strongly polynomial algorithm
for 2VPI that instead uses a label-correcting
approach reminiscent of the classical Bellman-Ford algorithm. However, this algorithm
runs in $O(m^2n^2)$ time, and thus is slower by at least a linear factor than all the known parametric
search-based methods.

\paragraph{Monotone 2VPI systems and graphs.}

It is well-known (see e.g.,~\cite{dadush2021accelerated, HochbaumMNT93}) that testing 2VPI feasibility can be reduced to the \emph{monotone} case (M2VPI)
where each inequality has at most one positive coefficient and at most one negative coefficient.
Let $(x)_{v\in V}$ be the variables, and let us index the inequalities
with the elements of a set~$E$ of size~$m$.
Then, after normalization, each inequality $e$ with variables $x_u,x_v$ (denoted $uv$) is of the form
\begin{equation}\label{eq:system}
  x_u\leq c(e)+\gamma(e)\cdot x_v \hspace{3mm}\text{for all }uv=e\in E
\end{equation}
where $c:E\to \mathbb{R}$ and $\gamma:E\to \mathbb{R}_{+}$. Clearly, such a system
of inequalities can be conveniently represented using a directed multigraph $G=(V,E)$
with $n$ vertices and $m$ edges.\footnote{Throughout (in particular in~(\ref{eq:system})), we use the notation $uv=e$ or $e=uv$ to pinpoint that
the tail and the head of the edge $e$ are $u$ and $v$ respectively. We stress that we allow many edges $e$ with the same tail and head and thus
$m$ might exceed $n^2$, e.g., $m$ might be exponential in $n$.}

M2VPI systems, if feasible, have a unique 
\emph{pointwise maximal} (\emph{pointwise minimal}) solution $\xmax{}=(\xmax{v})_{v\in V}$ ($\xmin{}=(\xmin{v})_{v\in V}$, respectively).
This means that for all feasible points $x$ we have $x_v\leq \xmax{v}$ ($x_v\geq \xmin{v}$, respectively) for every $v\in V$.
It might be the case that some $\xmax{v}$ equal $\infty$ (or $\xmin{v}$ equal $-\infty$) if the feasibility region is unbounded.

Interestingly, variants of some of the known 2VPI feasibility algorithms \linebreak (see e.g.~\cite{CohenM94, dadush2021accelerated})
also find the pointwise maximal (and symmetrically, pointwise minimal)
solution in the monotone case.

\paragraph{Deterministic Markov Decision Processes.}
Markov Decision Processes (MDPs) are a popular model for sequential decision making
under uncertainty. Given is a set of \emph{states} $V$ and a set of actions $E$
of the form $e=(s(e),p(e),c(e),\gamma(e))$, where $s(e)\in V$, $p(e)$ is a probability distribution
over $V$, $c(e)\in\mathbb{R}$, and $\gamma(e)\in (0,1)$ ($\gamma(e)$ is usually called a \emph{discount factor}). At a state $v\in V$, the \emph{controller} of an MDP
has to pick an action $e\in E$ with $s(e)=v$ (which is assumed to exist). Doing so, the controller gets an immediate reward $c(e)$.
Then, with probability $1-\gamma(e)$, the process stops. Otherwise, the controller gets transferred
to another state $w\in V$ with probability $p(e)(w)$.
The goal is to find a \emph{policy} (i.e., a strategy) for the controller that allows him to maximize the total
reward in this infinite stochastic process.
The optimal strategy for the MDPs is always deterministic and history independent~(see e.g.,~\cite{puterman1990}),
and thus can be described using a function~$\sigma$ mapping each state $v$ to one of the edges $e\in E$ with $s(e)=v$.
By solving an MDP we typically mean finding one such~$\sigma$. It is known that an MDP instance can be solved
using linear programming~\cite{d1963}.
Interestingly, strongly polynomial algorithms are known if the discount factors $\gamma(a)$ are all
equal to a constant $\gamma\in (0,1)$~\cite{HansenMZ13, Ye11a}.

In more restricted deterministic Markov Decision Processes (DMDPs), for each $a\in A$ the distribution $p(a)$
satisfies $p(a)(w)=1$ for some $w\in V$. As a result, the action determines the next state uniquely.
Such an MDP can be represented as a directed graph $G$ with a corresponding edge $e=uv$ for each action $e\in E$ with $s(e)=u$ and $p(e)(v)=1$.
The optimal policy $\sigma$ maps each vertex to one of its outgoing edges.
An optimal policy for a DMDP can be obtained by finding the pointwise maximal solution to the system~\eqref{eq:system},
which, since $\gamma(e)\in (0,1)$ for all $e\in E$, always exists and is bounded (for a proof, see e.g.,~\cite{MadaniTZ10}).

Currently, the most efficient known strongly polynomial algorithms for solving DMDPs in full generality
are those testing 2VPI feasibility~\cite{CohenM94, HochbaumN94} (see also~\cite{madani2000complexity}).
Less efficient strongly polynomial algorithms that avoid parametric search have also been described~\cite{dadush2021accelerated, HansenKZ14}.

The special case when all discount factors $\gamma(e)$ are uniform (equal to $\gamma$, where $\gamma$ can be a parameter,
not necessarily a constant) gained a lot of attention.
Madani~et~al.~\cite{MadaniTZ10} showed an $O(mn)$ time algorithm for this case.
Note that improving upon this bound would be a large breakthrough since $O(nm)$ remains the best
known strongly polynomial bound even for the simpler problem of negative cycle detection.
Weaker strongly polynomial upper bounds can be obtained by analyzing the performance of the simplex algorithm on such DMDPs~\cite{PostY15, HansenKZ14}.

\paragraph{Discounted APSP.}
Madani~et~al.~\cite{MadaniTZ10} also studied the following extension
of the DMDPs. Suppose the controller is allowed to end the process
after entering a fixed state $t\in V$, which is equivalent to adding a zero-cost loop action to $t$.
The \emph{discounted all-pairs shortest paths} problem asks to find an optimal policy for every possible
start/end state pair $s,t\in V$.

Besides of the motivation arising from DMDPs, the discounted APSP problem is interesting in its
own as a natural generalization of the standard APSP problem. In this generalization, the cost of a path 
$P=eQ$ for some $e\in E$ is given by $c(P)=c(e)+\gamma(e)\cdot c(Q)$ (for $\gamma(e)\leq 1$).
Note that in the standard APSP problem we have $\gamma(e)=1$ for all $e\in E$.

The discounted APSP problem turns out to be challenging because
the shortest paths do not possess the optimal substructure property anymore (except for the path suffixes)
and, as a result, the standard $\Ot(nm)$ or $O(n^3)$ time algorithms for APSP fail (see~\cite{MadaniTZ10}).
However, one can quite easily solve the problem in $O(mn^2)$ time by applying
a reverse Bellman-Ford algorithm $n$ times.

Madani~et~al.~\cite{MadaniTZ10} showed that despite the additional challenges,
in the \emph{uniform} case (when $\gamma(e)$ is equal for all $e$,
and thus the cost of a path $P=e_0\cdots e_k$ is given by $c(P)=\sum_{i=0}^k c(e)\cdot\gamma^i$),
the discounted APSP problem can be solved in $\Ot(n^2\sqrt{m})$ time with a randomized algorithm.
Note that this bound is $\Ot(n^3)$ for all $m$ and $\Ot(n^{2.5})$ for sparse graphs.\footnote{In the uniform case, it does not make sense to consider the case $m>n^2$.}
They asked whether there exists an $\Ot(nm)$ time algorithm for discounted APSP in the uniform case whose
performance would match, up to polylogarithmic factors, the known algorithms solving the standard APSP problem.

\subsection{Our results}

First of all, we show a strongly polynomial time-space trade-off algorithm for computing
a pointwise maximal solution to a monotone 2VPI system.

\begin{restatable}{theorem}{tvpi}\label{t:main}
  For any $h\in [1,n]$,
  there is a Las Vegas algorithm solving a monotone 2VPI system in $\Ot(mnh+(n/h)^3)$ expected time
  and $O(m)+\Ot(n^2/h)$ space. If the system is feasible, the algorithm computes the pointwise
  maximal solution. The time bound holds with high probability.
\end{restatable}
Recall that due to known reductions~\cite{dadush2021accelerated, HochbaumMNT93}, the algorithm behind Theorem~\ref{t:main} can be used
to solve the 2VPI feasibility problem in full generality and also to compute an optimal policy
for DMDPs.
In particular, by a suitable choice of the parameter $h$, we conclude the following.
\begin{restatable}{corollary}{specific}
There exist Las Vegas algorithms solving monotone 2VPI systems:
  \begin{enumerate}
    \item in $\Ot(nm+n^{3/2}m^{3/4})$ expected time using $\Ot(n^{3/2}m^{1/4})+O(m)$ space,
    \item in $\Ot(nm+n^3)$ expected time using $O(n+m)$ space.
  \end{enumerate}
\end{restatable}
The former bound improves upon the previous best randomized bound $\Ot(n^3+nm)$ of Cohen and Megiddo~\cite{CohenM94}
by a polynomial factor for non-dense instances with $m=O(n^{2-\eps})$, where $\eps>0$.
This answers an open problem of Madani~et~al.~\cite{MadaniTZ10}.\footnote{Who asked whether one can solve DMDPs in $o(mn^2)$ strongly polynomial time.
Whereas the randomized algorithm of Cohen and Megiddo~\cite{CohenM94} already improves upon $O(mn^2)$ time for most densities,
our algorithm is the first to break the $O(mn^2)$ bound for \emph{sparse} graphs.}
For such instances it also uses subquadratic space, whereas all the previous algorithms with $o(mn^3)$ time required
$\Theta(n^2+m)$ space.
The latter bound matches the previous best bound~\cite{CohenM94} but using \emph{linear} space as opposed
to quadratic in $n$.

We also give an improved algorithm for the discounted APSP problem.
\begin{restatable}{theorem}{discount}
  There exists a deterministic algorithm solving the discounted all-pairs
  shortest paths problem in the uniform case ($\gamma(e)=\gamma$ for all $e\in E$) in $\Ot(n^{3/2}m^{3/4})$ time.
\end{restatable}
Whereas we do not achieve the $\Ot(nm)$ bound in the uniform case, we improve the bound of Madani~et~al.~\cite{MadaniTZ10}
significantly for sparse graphs; then, our algorithm runs in $\Ot(n^{2+1/4})$ time.
Additionally, our algorithm is deterministic.

Unfortunately, neither our algorithm nor that of Madani~et~al.~\cite{MadaniTZ10} leads to an
$o(mn^2)$ time bound for discounted APSP problem in the non-uniform case.
Obtaining such an algorithm is an interesting open problem. The main obstacle
here is that in the case of varying discount factors it is even unclear
whether one can compute \emph{single-source discounted shortest paths} in $o(mn^2)$ time.\footnote{As opposed to \emph{single-target} shortest discounted paths, which can be computed in $O(mn)$ time using a Bellman-Ford-style procedure.}

\subsection{Technical overview}

\newcommand{\cycs}{\mathcal{C}}
\newcommand{\paths}{\mathcal{P}}
\newcommand{\bnds}{\mathcal{B}}
\newcommand{\ub}{\mathrm{ub}}
\newcommand{\cb}{\phi}

\paragraph{A simple randomized linear-space algorithm.}Our first contribution is a simple Las Vegas randomized algorithm
computing a pointwise maximal solution of a monotone 2VPI system in $O(mn^2\log{m})$ time and \emph{linear} space.
Note that this matches the best known deterministic time bound~\cite{HochbaumN94} which is
achieved using $\Theta(n^2+m)$ space.

Like all the previous polynomial algorithms designed specifically for 2VPI feasibility,
we make use of Shostak's characterization~\cite{Shostak81} of the pointwise maximal solution $\xmax{}$ (see Theorem~\ref{t:shostak}).
Roughly speaking, it says that for any $v\in V$, $\xmax{v}$ equals the minimum
upper bound on $x_v$ that can be devised from a simple path/simple cycle pair
$(P,C)$ such that $P=v\to w$ for some $w\in V$ and $C$ is a cycle in $G$ going through $w$
whose product $\gamma(C)$ of values $\gamma(e)$ is less than $1$.
It is easy to verify that such a cycle $C$ implies some absolute upper bound on $x_w$,
whereas such a path $P$ implies an upper bound on $x_v$ as a function of $x_w$.

Megiddo's algorithm~\cite{Megiddo83b} uses this characterization in the following way.
Let $\cb_w$ be the tightest upper bound that can be obtained using a simple cycle $C$
going through $w$. Megiddo shows that each \emph{cycle bound} $\cb_w$ can be computed in $O(mn^2\log{m})$ time and linear space
with parametric search.
This is possible since $\cb_w<\xi$ can be evaluated in $O(mn)$ time using a Bellman-Ford-style procedure of depth $O(n)$.\footnote{Here,
by the depth we mean the parallel time of the algorithm.}
Thus, computing the cycle bounds $\cb_w$ for all $n$ possible $w\in V$ takes $O(mn^3\log{m})$ time.
It is not difficult to show that $\xmax{}$ can be obtained by propagating
the bounds~$\cb_w$ using a Bellman-Ford-style procedure running in $O(nm)$ time and linear space.

On the other hand, the state-of-the-art algorithms for 2VPI feasibility~\cite{CohenM94, HochbaumN94} do not compute cycle bounds
at all. They instead use the following
key lemma attributed to Aspvall and Shiloach~\cite{AspvallS80}.\footnote{The attribution is made consistently throughout
the previous work, e.g., in~\cite{CohenM94, HochbaumN94}. Restrepo and Williamson~\cite{RestrepoW09} point
out that the procedure should be rather tracked to Aspvall's Ph.D. thesis~\cite{aspvall1981efficient}; they also argue that
the original implementation~\cite{aspvall1981efficient} uses $O(n^2)$ space, and provide a linear-space algorithm accomplishing
essentially the same task.}
\begin{lemma}\label{l:aspvall}
  Let $v\in V$ and $\xi\in \mathbb{R}$. One can check whether $\xmax{v}<\xi$ in $O(nm)$ time
  and using linear space.
\end{lemma}
Cohen and Megiddo~\cite{CohenM94} call the process behind Lemma~\ref{l:aspvall}
``locating'' a value $\xi$ with respect to~$\xmax{v}$.
Roughly speaking, the procedure behind Lemma~\ref{l:aspvall} looks
for a certificate, consisting of a path and a cycle, that imposes
an upper bound on $x_v$ tighter than $\xi$.

Hochbaum and Naor~\cite{HochbaumN94} use Lemma~\ref{l:aspvall} to eliminate a single variable $x_v$
by replacing it with up to $O(n^2)$ inequalities (computed via parametric search) 
between the neighbors of $v$ in $G$.
So, reducing the instance in terms of the number of variables comes at a cost of densifying
the instance.

The deterministic algorithm of Cohen and Megiddo~\cite{CohenM94} leverages Lemma~\ref{l:aspvall} while performing
a repeated-squaring-style all-pairs shortest paths
algorithm. This algorithm constructs paths of at most $2k$ edges that impose tightest inequalities
between all the $n^2$ pairs out of the $n^2$ tightest inequalities imposed by paths of length at most $k$. Picking
the right $n^2$ paths out of $n^3$ possibilities requires applying Lemma~\ref{l:aspvall}
for each variable once.

The above explains why for both these approaches~\cite{CohenM94, HochbaumN94}, 
the $\Theta(n^2)$ space usage seems inherent.
Moreover, the usage of Lemma~\ref{l:aspvall} (as opposed to computing cycle bounds,
which is a slightly easier task) seems essential as well.

To match the $\Ot(mn^2)$ time bound with a linear space algorithm, we retreat to the original
approach of Megiddo~\cite{Megiddo83b}. We observe that if some cycle $C$ at vertex $v$ with $\gamma(C)<1$ satisfies
$\cb_v=\xmax{v}$, we do not need to compute the cycle bound $\cb_w$ for all other vertices
$w\neq v$ of $C$. So this suggests that if $C$ is large, we can save a lot of time.
Moreover, if $\cb_v$ is attained using a cycle with $k$ edges, then the test $\cb_v<\xi$
can be performed correctly in $O(mk)$ time. 
As a result, parametric search can compute $\cb_v$ correctly in $O(mk^2\log{m})$ time 
as opposed to $O(mn^2\log{m})$.
Interestingly, this property of cycle bounds does not apply to certificates considered
in Lemma~\ref{l:aspvall}. 
Indeed, it is not clear whether one can prove a stronger version of Lemma~\ref{l:aspvall} that could test in $O(mk)$ time whether
there exists a
path-cycle pair of total size $k$ that certifies $\xmax{v}<\xi$.\footnote{The intuitive reason
why this may be hard is as follows: locating a value $\xi$ wrt. $\xmax{v}$ resembles finding a negative cycle
\emph{reachable} from $v$ in $G$. Whereas one can easily check whether there exists a negative cycle
of $k$ edges \emph{going through}~$v$ in $O(mk)$ time, it is widely believed that finding such
a small negative cycle in $G$ globally (which can be easily made reachable from an artificial source vertex with edges to all other vertices)
requires cubic time~\cite{WilliamsW18}.}

With these observations, the randomized approach is roughly as follows. For some $k$, consider
some $\leq n$ chosen cycles $C_v$ such that $|C_v|\in [k,2k]$ and $\xmax{v}$ is attained using the upper bound incurred by $C_v$
(possibly, in combination with some path from $v$ to $C_v$).
Let us denote by $V_k$ the set of such vertices $v$.
Suppose we sample $O((n/k)\log{n})$ vertices $s$, and compute for each sample $s$ the best cycle bound
on $x_s$ incurred by a cycle of length at most $2k$. Then, with high probability\footnote{That is, with probability at least $1-1/n^c$ for an arbitrarily chosen constant $c\geq 1$.}, for all $v\in V_k$
some sample $s_v$ will hit $C_v$ and we will thus correctly compute $\cb_{s_v}$ in $O(mk^2\log{n})$ time.
Hence, the total time through all $O((n/k)\log{n})$ samples $s_v$ will be $O(mnk\log^2{n})$.
Given all the cycle bounds $\cb_{s_v}$, we can compute the actual tightest upper bounds $\xmax{v}$ for all $v\in V_k$
in $O(nm)$ time by propagating them using a Bellman-Ford-style procedure.

By proceeding as sketched above for all $k=2^j$, $j=0,\ldots,O(\log{n})$, we will succeed in computing
$\xmax{v}$ for all $v\in V$ in $O(mn^2\log^2{n})$ time whp.
As we show, the sampling can be slightly refined so that the \emph{expected} time
of the algorithm is $O(mn^2\log{m})$, i.e., exactly matching the best known deterministic bound~\cite{HochbaumN94}.

\paragraph{The trade-off algorithm.} Our space-time trade-off is obtained by combining the described
linear-space algorithm with the $\Ot(n^3+nm)$ time algorithm of Cohen and Megiddo~\cite{CohenM94}.
Namely, for some threshold parameter $h$, we handle computing cycle bounds incurred
by short cycles of length less than $h$ in $\Ot(nmh)$ time using $O(\log{h})$ initial phases of the linear-space algorithm.
Next, we show that computing the cycle bounds incurred by long cycles of length $h$ or more
can be reduced, again via sampling, to a smaller (but necessarily dense) 2VPI instance with $\Ot(n/h)$ variables
and $\Ot((n/h)^2)$ inequalities. Such an instance can be solved in $\Ot((n/h)^3)$ time.
Finally, we show that this reduction can be performed in $\Ot(nmh)$ time and using $\Ot(n^2/h)$ space.
The reduction constitutes another application of parametric search. 
However, the key that enables its efficient implementation
is the main trick behind the randomized algorithm of Cohen and Megiddo~\cite{CohenM94}, i.e.,
an $\Ot(nm)$ expected time procedure that can run 
Lemma~\ref{l:aspvall} for every vertex $v\in V$ (critically, for each $v$ with a different threshold $\xi_v$) in parallel (see Lemma~\ref{l:cohen}).

We now note that whereas the $\Ot(nmh+(n/h)^3)$ time trade-off that we achieve has not
been described previously, we believe that it can also be achieved via a combination
of the randomized approach of Cohen and Megiddo~\cite{CohenM94} and the recent parallel APSP
algorithm of Karczmarz and Sankowski~\cite{KarczmarzS21}.
This is because Cohen and Megiddo's approach is, technically speaking, a randomized
reduction to parallel APSP in the comparison-addition model. Namely, if a parallel APSP algorithm
with work $W(n,m)$ bound and depth $d$ is provided, the running time
of the reduction is $\Ot(W(n,m)+d\cdot nm)$.
Since it is known~\cite{KarczmarzS21} that for any $d\in [1,n]$ one can achieve $W(n,m)=\Ot(nm+(n/d)^3)$,
the same randomized time bound as in Theorem~\ref{t:main} is obtained ultimately.
However, the space requirement remains $\Theta(n^2)$ if one proceeds this way,
so our approach is significantly more efficient in terms of space.

\paragraph{Certificates of infeasibility.} Shostak's characterization~\cite{Shostak81} (Theorem~\ref{t:shostak})
states that if a monotone 2VPI system is infeasible, then there exists a minimal \emph{certificate}
that consists of at most two cycles and at most one path in $G$.
In applications, e.g., in Wayne's generalized minimum cost flow algorithm~\cite{Wayne02}, simply declaring
a monotone 2VPI system infeasible (which can be easily done by checking if the obtained pointwise
maximal solution violates the system~\eqref{eq:system}) is not enough, and a certificate of infeasibility is required.
Moreover, computing such certificates is a bottleneck in Wayne's algorithm.
We observe that constructing a certificate can be a very subtle task and it can make
2VPI feasibility algorithms considerably more complicated (see Section~\ref{s:kcycle}).
Our simple randomized algorithm is capable of constructing a certificate of infeasibility,
and thus we improve the space requirement of Wayne's algorithm to linear when relying
on a $\Ot(mn^2)$ time 2VPI feasibility algorithm.
However, our trade-off algorithm can only provide such a certificate \emph{conditionally} on
whether the randomized algorithm of Cohen and Megiddo~\cite{CohenM94} has this property.
While Wayne~\cite[Section~6.2]{Wayne02} claims that all the previously known 2VPI feasibility algorithms have
this property, we are unsure this applies to~\cite{CohenM94}.
That being said, it is very likely that that (already complicated) algorithm can be suitably extended.
Then, our trade-off algorithm leads to a faster combinatorial generalized minimum cost flow algorithm.

\paragraph{Discounted all-pairs shortest paths.}
In our discounted APSP algorithm,
we build upon the observation of Madani~et~al.~\cite{MadaniTZ10}
that allows reducing the general discounted APSP problem to the case where one
only cares about \emph{simple and finite} paths.
Then, to obtain the improved algorithm for the uniform case,
we describe two \emph{different} approaches for reducing the number of sources
of interest.
More specifically, given a discounted shortest paths problem instance $\langle S,V\rangle$
that requires computing the discounted distances $\dist_G(s,t)$ for all $(s,t)\in S\times V$,
a \emph{sources set reduction step} reduces the problem $\langle S,V\rangle$
to a problem $\langle S',V\rangle$, where $|S'|\ll|S|$.

Both our sources set reduction steps (Lemma~\ref{l:red1}) are based on deterministically hitting sufficiently
long prefixes of discounted shortest paths from $S$ to $V$.
The first reduction leverages the suffix optimality property of discounted shortest paths (Observation~\ref{o:suffix})
whereas the other takes advantage of $k$-prefix optimality of discounted shortest paths (Observation~\ref{o:prefix}).
Starting with the problem $\langle V,V\rangle$, we end up applying
the first reduction step once, and the other $O(\log{n})$ times.

\section{Preliminaries}\label{s:prelims}

\paragraph{Paths, cycles, and walks.}
To avoid confusion we will use the standard term \emph{walk} to refer to a non-necessarily simple
path. A \emph{closed walk} is a walk that starts and ends in the same vertex.
Whenever we write \emph{path}, we actually mean a \emph{simple path} with no repeated vertices,
possibly except for its endpoints.
A \emph{cycle} is a closed simple path with equal source and target.
Sometimes when referring to a cycle we interpret it as a subgraph of $G$, and sometimes, when
a cycle has a distinguished source vertex $s$, we interpret it as an $s\to s$ path.
It should be clear from the context which interpretation we use.

Let $P=e_1\ldots e_k$ be an $s\to t$ walk in $G$. We define $|P|=k$.
The walk $P$ can be used to devise the following inequality:
\begin{equation*}
  x_s\leq c(e_1)+\gamma(e_1)\cdot (c(e_2)+\gamma(e_2)\cdot \left(\ldots \cdot (c(e_k)+\gamma(e_k)\cdot x_t)\right))=\sum_{i=1}^kc(e_i)\left(\prod_{j=1}^{i-1}\gamma(e_j)\right)+\prod_{i=1}^k\gamma(e_i)\cdot x_t.
\end{equation*}
Let us define $c(P)=\sum_{i=1}^kc(e_i)\left(\prod_{j=1}^{i-1}\gamma(e_j)\right)$ and $\gamma(P)=\prod_{i=1}^k\gamma(e_i)$.
Then we can write:
\begin{equation*}
  x_s\leq c(P)+\gamma(P)\cdot x_t.
\end{equation*}
Note that if $k=1$, i.e., $P$ consists of a single edge $e_1$, we have $c(P)=c(e_1)$ and $\gamma(P)=\gamma(e_1)$.
If~$P$ can be expressed as $P=P_1P_2$, then $c(P)=c(P_1)+\gamma(P_1)\cdot c(P_2)$ and $\gamma(P)=\gamma(P_1)\cdot \gamma(P_2)$.

If $C$ is an $s\to s$ closed walk, we obtain
the inequality $x_s\leq c(C)+\gamma(C)\cdot x_s$.
Depending on the position of $\gamma(C)$ wrt. $1$, $C$ can be used
to obtain upper or lower bounds on $x_s$.
If $\gamma(C)<1$, then~$C$ implies the inequality
$x_s\leq \frac{c(C)}{1-\gamma(C)}$.
If $\gamma(C)>1$, then it implies the inequality
$x_s\geq\frac{c(C)}{1-\gamma(C)}$.
If $\gamma(C)=1$, then the implied inequality is
$0\leq c(C)$. So, assuming $\gamma(C)=1$, if $c(C)<0$, then the system
is infeasible, and otherwise the obtained inequality is trivial.

A \emph{negative unit-gain} closed walk with $\gamma(C)=1$ and $c(C)<0$ constitutes a possible \emph{certificate of infeasibility}.
Another possible certificate of infeasibility
is a pair of $s\to s$ closed walks $C^\geq,C^\leq$ such that $\gamma(C^\geq)>1$, $\gamma(C^\leq)<1$
and $\frac{c(C^\geq)}{1-\gamma(C^\geq)}>\frac{c(C^\leq)}{1-\gamma(C^\leq)}$.
Then the lower- and upper bounds imposed by $C^\geq$ and $C^\leq$ on $x_s$ respectively are
contradictory.

Sometimes, for a cycle $C\subseteq G$ without a distinguished start vertex
we will write $c_s(C)$ to denote
the cost of~$C$ if $C$ is interpreted as an $s\to s$ path.

Denote by $\paths_{s,t}$ the set of \emph{simple} $s\to t$ paths in $G$.
In particular, let $\paths_{s,s}$ be the set of \emph{simple} cycles in $G$ that contain $s$.

For any $k\geq 0$, denote by $\paths_{s,t}^k$ the set of
$s\to t$ walks with at most $k$ edges.
Let $\cycs_v^k=\paths_{v,v}^k$ be the set of closed walks with at most $k$ edges and a distinguished
starting vertex $v$, so that for any $C\in\cycs_v^k$, $c(C)$ is well-defined.
For $C\in\cycs_v^k$, let $\cb(C)=\frac{c(C)}{1-\gamma(C)}$ be the \emph{cycle bound} of~$C$.
If $C$ has no distinguished source, and we interpret it as an $s\to s$ path,
we similarly set $\cb_s(C)=\frac{c_s(C)}{1-\gamma(C)}$.

Finally, note that if $C$ is a walk such that $C=PQ$, then:
\begin{equation}\label{eq:source-change}
  \cb(C)=\frac{c(C)}{1-\gamma(C)}=\frac{c(P)(1-\gamma(QP))+c(P)\gamma(QP)+\gamma(P)\cdot c(Q)}{1-\gamma(QP)}=c(P)+\gamma(P)\cdot \cb(QP).
\end{equation}

\paragraph{Feasibility of monotone 2VPI systems.} The following key result is attributed (see e.g.~\cite{CohenM94, dadush2021accelerated}) to Shostak~\cite{Shostak81} and
provides a way to test the feasibility of monotone 2VPI inequality systems.
Moreover, if the system is feasible, it gives explicit formulas for calculating the vectors $\xmax{},\xmin{}$.

\begin{theorem}\label{t:shostak}\textup{\cite{Shostak81}}
  Let:
  \begin{align}
    x_v^\leq&=\min_{w\in V} \min_{P\in \paths_{v,w}} \left\{c(P)+\gamma(P)\cdot \min_{\substack{C\in \paths_{w,w}\\\gamma(C)<1}}\left\{ \cb(C)\right\}\right\}\label{eq:shostak},\\
    x_v^\geq&=\max_{w\in V} \max_{P\in \paths_{w,v}} \left\{\frac{1}{\gamma(P)}\left(\max_{\substack{C\in \paths_{w,w}\\\gamma(C)>1}}\left\{\cb(C)\right\}-c(P)\right)\right\}.\nonumber
  \end{align}
  Then:
  \begin{enumerate}
    \item If the system~\eqref{eq:system} is feasible, $\xmax{v}=x_v^\leq$ and $\xmin{v}=x_v^\geq$ for all $v\in V$.
    \item If $x_v^\leq < x_v^\geq$ for some $v$, then the system is infeasible.
    \item If the system is infeasible, then there exists a certificate of infeasibility $H\subseteq G$, such
      that $H$ is either
      \begin{enumerate}[label=(\alph*)]
        \item a cycle with $\gamma(H)=1$ and $c(H)<0$ (so-called \emph{negative unit-gain cycle}), or
      \item $H=C^\leq \cup C^\geq\cup P$, where, for some possibly equal $s,t\in V$, $C^\leq\in\paths_{t,t}$,
          $C^\geq\in\paths_{s,s}$, $P\in\paths_{s,t}$, $\gamma(C^\leq)<1$, $\gamma(C^\geq)>1$, and
      $c_s(C^\geq)>c(P)+\gamma(P)\cdot c_s(C^\leq)$ (so-called \emph{negative bicycle}).
      \end{enumerate}
  \end{enumerate}
\end{theorem}
Theorem~\ref{t:shostak}
states that $\xmax{v}$ equals the tightest upper bound incurred on $v$ by a
path/cycle pair $(P,C)\in \bigcup_{w\in V}\paths_{v,w}\times \paths_{w,w}$.
Similarly, $\xmin{v}$  equals the tightest lower bound incurred on $v$ by a
cycle/path pair $(C,P)\in \bigcup_{w\in V}\paths_{w,w}\times \paths_{w,v}$.

\paragraph{Propagating bounds through walks.}
Suppose for all $v\in V$, $y_v\geq \xmax{v}$ is a valid upper bound on $x_v$.
By \emph{propagating bounds through edges} we mean the following operation.\footnote{Cohen and Megiddo~\cite{CohenM94} call this a \emph{push phase}.
Madani~\cite{madani2000complexity} calls this \emph{backward propagation}.}
First, record a copy of each $y_u$ as $y_u'$. For all $uv=e\in E$, set $y_u:=\min(y_u,c(e)+\gamma(e)\cdot y_v')$.
The operation clearly resembles a single phase of the Bellman-Ford algorithm.
Note that propagating bounds through edges keeps the upper bounds valid since
$x_v\leq y_v$ along with $x_u\leq c(e)+\gamma(e)x_v$ implies
that $c(e)+\gamma(e)y_v$ is also an upper bound on $x_u$.

More generally, we call applying $k\geq 1$ propagation steps
\emph{propagating bounds through walks of length~$k$}.
It is easy to see that after propagating bounds through walks of length $k$,
the obtained upper bounds $y_s$ (for $s\in V$) satisfy
$y_s=\min_{P\in\paths_{s,t}^k}\{c(P)+\gamma(P)\cdot y_t'\}$ for any $t\in V$,
where~$y_t'$ is the upper bound on $x_t$ \emph{before} the propagation.
Note that propagating through walks of length $k$ takes $O(mk)$ time
and requires linear additional space.

\section{Simplifying the problem}\label{s:simplify}
First of all, note that any closed walk $C\in \bigcup_{k=1}^\infty\cycs_v^k$ with $c(C)<0$ and $\gamma(C)=1$ certifies the infeasibility
of the system. Similarly, any triple $C^\leq\in \bigcup_{k=1}^\infty\paths_{t,t}^k$, $C^\geq\in\bigcup_{k=1}^\infty\paths_{s,s}^k$, $P\in\bigcup_{k=0}^\infty\paths_{s,t}^k$
such that $c(C^\leq)<1$, $c(C^\geq)>1$, and $c(C^\geq)>c(P)+\gamma(P)\cdot c(C^\leq)$ certifies infeasibility of the system.
We refer to such certificates as non-simple, whether to those described in Theorem~\ref{t:shostak} (item~3) as simple.

It is not very difficult to check that any non-simple certificate can be actually ``simplified'' in linear time
by removing some redundant cycles (however, e.g., a non-simple negative cycle with $\gamma(C)=1$ might
get simplified into a negative bicycle).
That being said, in the following, we will only care about producing non-simple certificates.

\newcommand{\vbad}{v_{\mathrm{bad}}}
For any $v\in V$, fix $(P_v,C_v)$ to be any simple path/simple cycle pair minimizing the right-hand side
of~\eqref{eq:shostak}, if $x_v^\leq\neq\infty$.
We now define a set $\cycs^*$ of \emph{distinguished cycles}.
If the system~\eqref{eq:system} is feasible, we 
set $\cycs^*$ to be any \emph{maximal} subset of vertex-disjoint cycles of $\{C_v:v\in V,x^\leq_v\neq\infty \}$.
Otherwise, if it is infeasible, let $\vbad\in V$ be such that $x_{\vbad}^\leq<x_{\vbad}^\geq$.
Then $\cycs^*=\{C_{\vbad}\}$. Note that such $\vbad$ exists by Theorem~\ref{t:shostak}.
Clearly, by the definition, we have $\sum_{C\in\cycs^*}|C|\leq n$ in any case.
In the following we assume that the distinguished cycles do not have a specified starting point.

\begin{definition}
  Let $v\in V$. 
  For $k\geq 1$, we call the value
  \begin{equation*}
    \cb_{v,k}=\min\left\{\cb(C):C\in\cycs_v^k, \gamma(C)<1\right\}
  \end{equation*}
  the \emph{$k$-cycle upper bound} of $v$ incurred by $G$.
\end{definition}
We start with the following technical lemma.
\begin{restatable}{lemma}{lred}\label{l:reduction}
  For each $v\in V$ let $x^*_v\in \mathbb{R}\cup\{\infty\}$ be values satisfying the following:
  if $x^*_v\neq\infty$, then there exists $F_v\in\paths_{v,f}^n$ and $D_v\in\cycs_f^n$ for some $f\in V$ such that $\gamma(D_v)<1$ and $x^*_v=c(F_v)+\gamma(F_v)\cdot \cb(D_v)$.
  Moreover, suppose that for each $C\in\cycs^*$ there exists such a $w\in V(C)$
  that $x^*_w\leq \cb_{w,|C|}$.
  Let
  \begin{equation}\label{eq:ystar}
    y^*_v:=\min_{w\in V}\min_{P\in\paths_{v,w}^{3n}}\left\{c(P)+\gamma(P)\cdot x^*_w\right\}.
  \end{equation}
  If the system~\eqref{eq:system} is feasible, then for any $v\in V$, $y^*_v=\xmax{v}$.
  Otherwise, $y^*_{\vbad}\leq x^\leq_{\vbad}$ holds.

  The values $y^*_v$, along with endpoints $w_v$ of some arbitrary corresponding walks $P=v\to w_v$ minimizing~\eqref{eq:ystar}, can be found in $O(nm)$ time and linear space.
\end{restatable}
\begin{proof}
  First suppose the system is feasible. Then note that in the right hand side of~\eqref{eq:shostak}
  we can replace the set $\paths_{v,w}$ over which $P$ is minimized with $\paths_{v,w}^{4n}$.
  This cannot decrease $x_v^\leq$ since it would then make the vector $\xmax{}$ infeasible, which would
  contradict Theorem~\ref{t:shostak}.
  Similarly, we can replace $P_{w,w}$ with $\cycs_w^{n}$ in that formula.
  Since each finite $x_v^*$ is an upper bound on $x_v$ achieved using a walk of no more than $n$ edges
  and a closed walk of no more than $n$ edges,
  each $y_v^*$, if finite, is an upper bound on $x_v$ achieved using a walk with no more
  than $3n$ edges and a closed walk with no more than $n$ edges.
  It follows that $x_v^\leq\leq y^*_v$.

  Recall that we have fixed $P_v,C_v$ for any $v$ to be some simple path/simple cycle minimizing the right-hand side
  of~\eqref{eq:shostak} when defining the distinguished cycles $\cycs^*$.
  We first argue that if some $z\in V$ lies on both $C_v$ and $C_w$ for $v\neq w$,
  then $\cb_z(C_v)=\cb_z(C_w)$. Otherwise, we would wlog. have $\cb_z(C_w)<\cb_z(C_v)$.
  Then, if $P_v=v\to w$, then for the $v\to z$ subpath $Q$ of $C_v$
  we would have
  \begin{align*}
    x_v^\leq&\leq c(P_vQ)+\gamma(P_vQ)\cdot \cb_z(C_w)\\
            &<c(P_vQ)+\gamma(P_vQ)\cdot \cb_z(C_v)\\
            &=c(P_v)+\gamma(P_v)\cdot(c(Q)+\gamma(Q)\cdot \cb_z(C_v))\\
            &=c(P_v)+\gamma(P_v)\cdot \cb_v(C_v)\\
            &=x_v^\leq,
  \end{align*}
  where in the penultimate step we use Equation~\eqref{eq:source-change}. This is clearly a contradiction.

  Now fix some $v\in V$.
  Let $C_v'\in \cycs^*$ be such that there exists a vertex $s\in V(C_v)\cap V(C_v')$.
  Such a $C_v'$ exists since $\cycs^*$ is maximal.
  Moreover, for some $t\in V(C_v')$ we have
  \begin{equation*}
    c(F_t)+\gamma(F_t)\cdot \cb(D_t)=x^*_t\leq \cb_{t,|C_v'|}\leq \cb_t(C_v'),
  \end{equation*}
  where $\gamma(D_t)<1$, $F_t\in\paths_{t,f}^n$, and $D_t\in \cycs_f^n$ for some $f\in V$.
  For $P_v=v\to w$, let $Q$ be the $w\to s$ subpath of the cycle $C_v$, and let $R$ be the $s\to t$ subpath
  of the cycle $C_v'$.
  We obtain:
  \begin{align*}
    x_v^\leq &= c(P_v)+\gamma(P_v)\cdot \cb_v(C_v)\\
             &=c(P_v)+\gamma(P_v)\cdot (c(Q)+\gamma(Q)\cdot \cb_s(C_v))\\
             &=c(P_v)+\gamma(P_v)\cdot (c(Q)+\gamma(Q)\cdot \cb_s(C_v'))\\
             &=c(P_v)+\gamma(P_v)\cdot (c(Q)+\gamma(Q)\cdot (c(R)+\gamma(R)\cdot \cb_t(C_v')))\\
             &=c(P_vQR)+\gamma(P_vQR)\cdot \cb_t(C_v').
  \end{align*}
  We also have $x_v^\leq\leq c(P_vQRF_t)+\gamma(P_vQRF_t)\cdot \cb(D_t)=c(P_vQR)+\gamma(P_vQR)\cdot x^*_t$
  since, as discussed before, in~\eqref{eq:shostak} we can equivalently minimize over pairs of walks of length no more than $4n$
  and closed walks of length at most $n$.
  So we obtain:
  \begin{equation*}
    x_v^\leq =c(P_vQR)+\gamma(P_vQR)\cdot \cb_t(C_v')\geq c(P_vQR)+\gamma(P_vQR)\cdot x^*_t\geq x_v^\leq.
  \end{equation*}
  It follows that $x_v^\leq=c(P_vQR)+\gamma(P_vQR)\cdot x^*_t$.
  Finally, since $P_vQR\in \paths_{v,t}^{3n}$, $y^*_v\leq x_v^\leq$.
  
  We have proved that if the system is feasible, we indeed have $y^*_v=x_v^\leq$  for all $v\in V$.
  By Theorem~\ref{t:shostak}, it follows that $y^*_v=\xmax{v}$ for all $v\in V$.

  Now assume that the system is infeasible. Let $P_{\vbad}=\vbad\to w$.
  Then there exists $t\in V(C_{\vbad})$ such that if $Q$ is the $w\to t$
  path of $C_{\vbad}$, then
  \begin{align*}
    x_{\vbad}^\leq&=c(P_{\vbad})+\gamma(P_{\vbad})\cdot \cb_w(C_{\vbad})\\
                  &=c(P_{\vbad})+\gamma(P_{\vbad})\cdot (c(Q)+\gamma(Q)\cdot \cb_t(C_{\vbad}))\\
                  &\geq c(P_{\vbad})+\gamma(P_{\vbad})\cdot (c(Q)+\gamma(Q)\cdot x^*_t)\\
                  &=c(P_{\vbad}Q)+\gamma(P_{\vbad}Q)\cdot x^*_t\\
                  &\geq y^*_t,
  \end{align*}
  as clearly $P_{\vbad}Q\in \paths_{v,t}^{3n}$ holds.

  The algorithm to compute the ultimate values $y_v^*$ is to simply set $y_v^*:=x_v$ initially and then propagate these values through paths
  of length no more than~$3n$. This takes $O(nm)$ time. Propagation can be easily extended to store
  an endpoint $w_v$ of some path $P$ such that $y_v^*=c(P)+\gamma(P)\cdot x^*_w$.
  \end{proof}

We now argue that in $O(mn\log{n})$ time we can reduce solving monotone 2VPI
feasibility to the following more restricted problem.
\begin{problem}\label{prob}
  Compute a vector of values $x_v^*$, $v\in V$, satisfying:
  \begin{itemize}
    \item If $x_v^*\neq\infty$, then $x_v^*=c(F)+\gamma(F)\cdot \cb(D)$ for some $F\in\paths_{v,w}^n$, $D\in\cycs_w^n$, $\gamma(D)<1$.
    \item For each $C\in\cycs^*$, there exists a vertex $s\in V(C)$ such that $x_s^*\leq \cb_{s,|C|}$.
  \end{itemize}
\end{problem}
\begin{remark}\label{r:prob}
  If we seek a certifying monotone 2VPI feasibility algorithm, then 
  the algorithm solving Problem~\ref{prob} has to possess the following additional properties:
  \begin{itemize}
    \item The algorithm is allowed to fail but then a certificate of infeasibility of the system
  is required.
    \item After the values $x_v^*$ are successfully computed,
  the algorithm may be requested to compute a corresponding walk/closed walk pair $(F,D)$
  as defined above
  for some \emph{single given} $v$ with $x_v^*\neq\infty$.
  \end{itemize}
\end{remark}

First, suppose the system is feasible. Then, indeed Lemma~\ref{l:reduction}
implies that the pointwise maximal solution can be computed
after solving Problem~\ref{prob}.
By symmetry in Theorem~\ref{t:shostak}, the values $\xmin{v}$ can be computed
by finding the pointwise maximal solution for a ``reverse'' instance $G^R$, where each
edge $e$ is reversed, has cost $c^R(e)=-c(e)/\gamma(e)$, and $\gamma^R(e)=1/\gamma(e)$.

If the system is infeasible and while solving Problem~\ref{prob} we do not
encounter a certificate of infeasibility, then the values $y^*_v$ from Lemma~\ref{l:reduction} (whose definition
does not depend on feasibility) are computed correctly. The same applies
to the analogously defined ``reverse'' values $(y^*_v)^R$.
Recall that if the system is infeasible then for some~$v$ we have
$x^\leq_{\vbad}<x^\geq_{\vbad}$.
Lemma~\ref{l:reduction} guarantees that $y^*_{\vbad}\leq x^\leq_{\vbad}$ in such case.
The same applies to the value $(y^*_{\vbad})^R$, i.e., $x^\geq_{\vbad}\leq (y^*_{\vbad})^R$.
As a result, we have $y^*_{\vbad}<(y^*_{\vbad})^R$.
We do not know $\vbad$ beforehand, but its existence proves
that $y^*_{g}<(y^*_{g})^R$ for some $g\in V$.
Clearly, we can find one such $g$ in $O(n)$ time.

Since $y_g^*=\min_{w\in V}\min_{P\in\paths_{g,w}^{3n}}\left\{c(P)+\gamma(P)\cdot x^*_w\right\}$,
the following lemma,
proved in Section~\ref{s:construct-path},
can be used to produce such a walk $S_g=g\to w_g$
that $y^*_{g}=c(S_g)+\gamma(S_g)\cdot x^*_{w_g}$.
\begin{restatable}{lemma}{lconstructpath}\label{l:construct-path}
  Let $s,t\in V$ and $\alpha\in\mathbb{R}$. Then, an $s\to t$ walk $Q\in\paths_{s,t}^k$
  lexicographically minimizing the pair $(c(Q)+\gamma(Q)\cdot \alpha,\gamma(Q))$ 
  can be constructed in $O(mk\log{k})$ time and linear space.
\end{restatable}
We can thus request from the algorithm solving Problem~\ref{prob}
a walk/closed walk pair $F,D$ with $\gamma(D)<1$ such that $x^*_{w_g}=c(F)+\gamma(F)\cdot\cb(D)$.
Symmetrically, we can construct a walk $S_g^R=u_g\to g$, and a pair $F^R,D^R$ 
(note that $S_g^R, F^R,D^R$ are walks in $G$, not $G^R$)
such that $(y^*_g)^R=\left((\cb(D^R)-c(F^R))/\gamma(F^R)\right)-c(S_g^R))/\gamma(S_g^R)$.
Observe that then the closed walks $D,D^R$ along with the walk $F^R\cdot S_g^R\cdot S_g\cdot F$ form a negative bicycle
that constitutes a certificate of infeasibility.

\section{A simple randomized linear-space algorithm}\label{s:simple}

In this section, we describe a simple randomized algorithm 
solving Problem~\ref{prob}. 
The following lemma gives a method for computing values $\cb_{v,k}$ for any $k$.
The proof is quite technical and and is deferred to Section~\ref{s:kcycle}.
\begin{restatable}{lemma}{lparametric}\label{l:parametric}
  Let $v\in V$ and $k\leq n$.
  Then, in $O(mk^2\log{m})$ time and linear space one can
  compute either the value
  $\cb_{v,k}$ (along with a closed walk $C\in\cycs_v^k$ such that $\gamma(C)<1$ and $\cb(C)=\cb_{v,k}$
  if $\cb_{v,k}\neq\infty$)
  or a set of at most two closed walks certifying the infeasibility of the system.
\end{restatable}
Note that Problem~\ref{prob} can be easily solved deterministically
in $O(mn^3\log{m})$ time by setting $x^*_v=\cb_{v,n}$ for each $v\in V$.
However, for example, for large distinguished cycles $C\in\cycs^*$
it would be enough to set $x^*_u=\cb_{u,|C|}$ for just one vertex $u\in V(C)$,
instead of all $|C|$ vertices. Unfortunately, it is not clear how to
construct $\cycs^*$ deterministically without computing $\cb_{v,n}$ for all $v\in V$,
which again takes $O(mn^3\log{m})$ time.
We will thus use random sampling to hit all the distinguished cycles.
For example, to hit all distinguished cycles of length $\Theta(n)$ with high probability,
only $\Ot(1)$ samples are needed.

The simple randomized algorithm proceeds in $O(\log{n})$ phases numbered $0,\ldots,\ell$, where $\ell$ is the largest
integer such that $2^\ell\leq n$.
Initially, $x_v^*=\infty$ for all $v$. Consider the $j$-th phase. In the $j$-th phase we repeat the following
$t_j=\lceil n/2^j\rceil \cdot (\ell+2-j)^3$ times. First, sample a vertex $v\in V$ uniformly at random.
Then, compute $\cb_{v,2^{j+1}}$ using Lemma~\ref{l:parametric} and set $x_v^*:=\min(x_v^*,\cb_{v,2^{j+1}})$.
If Lemma~\ref{l:parametric} fails with a certificate of infeasibility, we terminate
with that certificate.
\begin{lemma}\label{l:simple-correct}
  If no certificate of infeasibility is produced, then
  with constant probability, for each distinguished cycle $C\in\cycs^*$ there exists $v\in V(C)$
  such that $\cb(D_v)\leq x_v^*\leq \cb_{v,|C|}$ for some $D_v\in\cycs_v^n$.
\end{lemma}
\begin{proof}
Denote by $p$ the probability that we fail to achieve the goal from the lemma's statement.
So our aim is on proving that $p$ is bounded from above by a constant less than $1$.

  Let $\cycs^*_j\subseteq \cycs^*$ be the subset of cycles of $\cycs^*$ of size between $2^j$ and $2^{j+1}$.
  Since the cycles in $\cycs^*$ are vertex-disjoint, we have $\sum_{C\in\cycs^*}|C|\leq n$.
  As a result, $|\cycs^*_j|\leq n/2^j$.
 
  In the following, \emph{to hit a cycle $C\in\cycs^*_j$} means to sample a vertex that is an element of $V(C)$.
  Clearly, after a cycle is hit with $v\in V(C)$, we set $x^*_v$ to a value no more than $\cb_{v,2^{j+1}}\leq \cb_{v,|C|}$.
  Let $p_j$ be the probability that during phase $j$ we fail to hit some $C\in \cycs^*_j$.
  By the union bound, we have:
  \begin{equation*}
    p\leq \sum_{j=0}^\ell p_j.
  \end{equation*}

  Imagine that we sample vertices in the $j$-th phase indefinitely.
  For any $S\subseteq \cycs^*_j$, let
  the random variable $X_{j,S}$ denote the minimal number of initial
  samples in the $j$-th phase that hit every cycle in $S$ at least once.

  We now show that for any $S\subseteq \cycs^*_j$, $\mathbb{E}[X_{j,S}]\leq \frac{n}{2^j}\cdot \sum_{i=1}^{|S|} \frac{1}{i}$.
  We prove the claim by induction on~$|S|$. The case $S=\emptyset$ is trivial.
    For any $C\in \cycs^*_j$, the probability that a sample hits the cycle $C$ is clearly $|C|/n$. 
    Moreover, let $y=\sum_{C\in S}|C|$.
    We have:
    \begin{align*}
      \mathbb{E}[X_{j,S}]&=1+\left(\sum_{C\in S}\frac{|C|}{n}\cdot \mathbb{E}[{X_{j,S\setminus\{C\}}}]\right)+\frac{n-y}{n}\cdot \mathbb{E}[X_{j,S}]\\
                        &\leq1+\left(\sum_{C\in S}\frac{|C|}{n}\cdot \left(\frac{n}{2^j}\cdot \sum_{i=1}^{|S|-1}\frac{1}{i}\right)\right)+\frac{n-y}{n}\cdot \mathbb{E}[X_{j,S}]\\
                        &\leq 1+\frac{y}{2^j}\cdot \left(\sum_{i=1}^{|S|-1}\frac{1}{i}\right)+\left(1-\frac{y}{n}\right)\cdot \mathbb{E}[X_{j,S}].
    \end{align*}
    By a simple transformation, and since $y\geq |S|\cdot 2^j$, we obtain:
    \begin{equation*}
      \mathbb{E}[X_{j,S}]\leq \frac{n}{y}+\frac{n}{2^j}\cdot \sum_{i=1}^{|S|-1}\frac{1}{i}\leq \frac{n}{2^j}\cdot \frac{1}{|S|}+\frac{n}{2^j}\cdot \sum_{i=1}^{|S|-1|}\frac{1}{i}=\frac{n}{2^j}\cdot \sum_{i-1}^{|S|}\frac{1}{i},
    \end{equation*}
    as desired.
  By the obtained bound and $|\cycs^*_j|\leq n/2^j$, we have:
  \begin{equation*}
    \mathbb{E}\left[X_{j,\cycs^*_j}\right]\leq \frac{n}{2^j}\cdot \sum_{i=1}^{|\cycs^*_j|}\frac{1}{i}\leq 
    \frac{n}{2^j}\left(\ln\frac{n}{2^j}+1\right)\leq \frac{n}{2^j}\cdot \log_2{\frac{2n}{2^j}}\leq \frac{n}{2^j}(\ell+2-j).
  \end{equation*}
  Since in the $j$-th phase we perform sampling $t_j=\lceil n/2^j\rceil \cdot (\ell+2-j)^3$ times, we have:
  \begin{equation*}
    p_j\leq \mathbb{P}\left[X_{j,\cycs^*_j}>t_j\right]\leq
    \mathbb{P}\left[X_{j,\cycs^*_j}>\mathbb{E}\left[X_{j,\cycs^*_j}\right]\cdot (\ell+2-j)^2\right]\leq \frac{1}{(\ell+2-j)^2},
  \end{equation*}
  where the latter bound above follows by Markov's inequality.
  Consequently, we have:
  \begin{equation*}
    p\leq \sum_{j=0}^\ell p_j\leq \sum_{j=0}^\ell \frac{1}{(\ell+2-j)^2}\leq \left(\sum_{j=1}^\infty\frac{1}{i^2}\right)-1=\frac{\pi^2}{6}-1<\frac{2}{3},
  \end{equation*}
  which concludes the proof.
\end{proof}
\begin{lemma}\label{l:simple-time}
  The $\ell+1$ phases of the algorithm require $O(mn^2\log{m})$ time and linear space.
\end{lemma}
\begin{proof}
  By Lemma~\ref{l:parametric}, the running time can be bounded as follows:
  \begin{equation*}
    \sum_{j=0}^\ell \left(O\left(m\cdot (2^j)^2 \log{m}\right)\cdot \left\lceil\frac{n}{2^j}\right\rceil\cdot (\ell+2-j)^3\right)=O\left(nm\log{m}\cdot\sum_{j=0}^\ell 2^j(\ell+2-j)^3\right).
  \end{equation*}
  Moreover, we have:
  \begin{equation*}
    \sum_{j=0}^\ell 2^j(\ell-j+1)^3=2^{\ell+2}\cdot \sum_{j=0}^\ell\frac{(\ell+2-j)^3}{2^{\ell+2-j}}\leq 2^{\ell+2}\cdot \sum_{j=0}^\infty\frac{j^3}{2^j}=2^{\ell+2}\cdot O(1)=O(n).
  \end{equation*}
\end{proof}
\begin{remark}
  The purpose of using $\lceil n/2^j\rceil \cdot (\ell+2-j)^3$ samples instead of simply e.g., $\Theta(\lceil n/2^j\rceil\cdot\ell)$
  (which would guarantee high probability correctness) in the $j$-th phase
  is to save an $O(\log n)$ factor in the final expected time bound.
\end{remark}

\begin{lemma}\label{l:simple}
  There is a Las Vegas algorithm solving a monotone 2VPI system (or producing a minimal certificate
  of infeasibility) in $O(mn^2\log{m})$ expected time
  and linear space.
  With high probability, the algorithm finishes in $O(mn^2\cdot \log{m}\cdot \log{n})$ time.
\end{lemma}
\begin{proof}
  The lemmas~\ref{l:simple-correct}~and~\ref{l:simple-time} combined show that $\ell+1$ phases produce a correct solution or certify infeasibility with probability at least $1/3$
  in $O(mn^2\log{m})$ time.  

  We now discuss how to verify whether the algorithm made a mistake and we need to rerun it.
  Consider some run of the algorithm.
  If the run failed with a certificate, then we can stop since the system is infeasible for sure.
  Otherwise, the run attempted to compute a pointwise maximal solution $x$.
  If $x$ does not satisfy~\eqref{eq:system}, which can be tested in $O(m)$ time, then the run failed to produce a certificate, 
  so we restart.
  Otherwise, $x$ certifies that the system is feasible, so we only need to verify whether $x$ is indeed pointwise maximal.
  Clearly, if for some $u$ we have $x_u<\min_{uv=e\in E}\{c(e)+\gamma(e)\cdot x_v\}$, then
  one can increase $x_u$ slightly without violating the inequalities and obtain a better solution.
  So the run indeed made a mistake.
  Otherwise, $x$ is a vertex of the feasible polytope.
  So if $x$ is not optimal, it is possible to move to an adjacent vertex which improves
  the objective function $\sum_{v\in V}x_v$.
  Therefore, for each $uv=e\in E$ with $x_u=c(e)+\gamma(e)\cdot x_v$ we need
  to check whether for some $\alpha>0$ one can increase $x_u$ by $\alpha$ and $x_v$ by $\gamma(e)\alpha$
  without violating any of the other inequalities.
  This can be easily verified in $O(nm)$ time or even near-linear time.

  Clearly, the expected number of times the algorithm is repeated before producing
  a correct answer is at most~$3$.
  Within $O(\log{n})$ repeats, the failure probability gets decreased to inversely polynomial in $n$.

  Since we only store the values $x_v^*$ for $v\in V$, the used space is linear by Lemma~\ref{l:parametric}.
\end{proof}

\section{A trade-off algorithm}

The randomized trade-off algorithm solving Problem~\ref{prob} is obtained via a combination of the simple
approach 
of Section~\ref{s:simple}
with the following result of Cohen and Megiddo~\cite{CohenM94}.
\begin{theorem}\label{t:cohen}
  There exists a Las Vegas randomized algorithm testing feasibility and computing a pointwise maximal solution
  to a monotone 2VPI system~\eqref{eq:system} in $\Ot(n^3+nm)$ time and using $O(n^2+m)$ space.
  The time bound holds with high probability.
\end{theorem}

This combination is made possible by the key technical trick from~\cite{CohenM94} that we state as the following lemma.
\begin{lemma}\label{l:cohen}\textup{\cite[Proposition~5.2]{CohenM94}}
  Suppose the system~\eqref{eq:system} is feasible.
  Let $\xi_v\in \mathbb{R}$ for each $v\in V$. 
  In $\Ot(nm)$ expected time and linear space one can decide, for each individual $v\in V$, whether
  $\xmax{v}<\xi_v$.
  The algorithm is Las Vegas randomized and the time bound holds with high probability.
\end{lemma}

Note the difference between Lemma~\ref{l:cohen} and Lemma~\ref{l:aspvall}.
The latter can detect whether for \emph{some}~$v$ we have $\xmax{v}<\xi_v$.
Lemma~\ref{l:cohen} can decide this for \emph{all} $v\in V$ at once
within roughly the same time. However, this comes at a cost of using
randomization that seems difficult to eliminate.
Since the algorithm in~\cite{CohenM94} actually required $O(n^2)$ space, we will next sketch this algorithm
and explain how the space can be reduced to linear by 
using the following generalization of Lemma~\ref{l:aspvall}.
\begin{restatable}{lemma}{locateext}\label{l:locate-ext}
  Suppose the system~\eqref{eq:system} is feasible.
  Let $\xi_v\in \mathbb{R}\cup\{-\infty\}$ for each $v\in V$.
  
  In $O(nm)$ time and linear space one can decide whether $\xmax{v}<\xi_v$ holds for \emph{some} $v\in V$.
  Moreover, the algorithm has the following properties.
  \begin{itemize}
    \item If the answer is affirmative, then for some $v$ with $\xmax{v}<\xi_v$, 
      a certificate $(P,C)\in \paths_{v,w}^n\times \cycs_{w,w}^n$ such that
      $\gamma(C)<1$ and
      $c(P)+\gamma(P)\cdot \cb(C)<\xi_v$ is returned.
      Moreover, in such a case, if $|V(P\cup C)|=\ell$, then the algorithm
      terminates in $O(m\ell)$ time.\footnote{We stress that it \emph{does not} necessarily hold that if some certificate $(P',C')$ with $|V(P'\cup C')|=\ell$
      exists, then the algorithm terminates in $O(m\ell)$ time. In such a case, the algorithm may still take e.g., $\Theta(nm)$ time, but then
      the returned certificate $(P,C)$ satisfies $|V(P\cup C)|=\Theta(n)$.}

    \item If $\xmax{v}<\xi_v$ holds for precisely one $v\in V$, then
      the algorithm produces the same certificate regardless of the exact values
      of $\xi_w\leq \xmax{w}$ for all other $w\in V\setminus\{v\}$.
 \end{itemize}
\end{restatable}
Lemma~\ref{l:locate-ext} can be obtained by
adapting and analyzing more carefully
the negative-cost generalized augmenting path detection algorithm
of Restrepo and Williamson~\cite[Section~4]{RestrepoW09}.
The proof of Lemma~\ref{l:locate-ext} can be found in Section~\ref{s:locateext}.
Cohen and Megiddo~\cite[Algorithm~2.18]{CohenM94} provide a location procedure achieving almost the same guarantees
as in Lemma~\ref{l:locate-ext}, except that their procedure requires $O(n^2)$ space.

\begin{proof}[Sketch of the algorithm behind Lemma~\ref{l:cohen}]
  Let us fix for each $v\in V$ with $\xmax{v}<\xi_v$ a certificate $(P_v,C_v)$ that
  Lemma~\ref{l:locate-ext} produces if run with a vector $\xi'$ such that
  $\xi'_v=\xi_v$ and $\xmax{w}\geq \xi'_w$ for all $w\neq v$.
  By Lemma~\ref{l:locate-ext}, this certificate is defined uniquely and
  regardless of the values $\xi'_w$ for $w\neq v$.

  Given $\xi$, let us call a vertex $v$ \emph{big} if $\xmax{v}<\xi_v$. Let us call a big vertex $\ell$-big if $|V(P_v\cup C_v)|=\Theta(\ell)$.
  
  Roughly speaking, the algorithm in~\cite{CohenM94} does the following.
  Let $I\subseteq V$ be the (unknown) set of big vertices that have not yet been determined,
  and let $D$ be the set of big vertices that we have detected.
  Using Lemma~\ref{l:locate-ext} one can check if $O(mn)$ time whether $I=\emptyset$ and
  we should stop.
  More generally, using random sampling and Lemma~\ref{l:locate-ext}, in $\Ot(mn)$ time one can compute an estimate $g$ approximating
  $|I|$ up to a constant factor.

  If $I\neq\emptyset$, then based on $g$ one can sample $O(\polylog{n})$ sets $S_1,\ldots,S_k\subseteq V\setminus D$
  such that with high probability, at least $\Omega(\log{n})$ of the samples $S_i$ contain a unique big vertex $v$,
  that can be additionally assumed to be randomly picked from $I$.
  Consider running the algorithm of Lemma~\ref{l:locate-ext} with
  a vector $\xi'$ such that $\xi'_s=\xi_s$ for $s\in S_i$ and $\xi'_s=-\infty$
  for $s\in V\setminus S_i$ as an input.
  If $v$ is $\ell$-big, then the run will terminate in $O(m\ell)$ time.
  Moreover, one can prove that if $v$ is randomly picked from $I$,
  then by propagating the upper bounds following from the certificate $(P_v,C_v)$
  through walks of length $\Theta(\ell)$, with constant probability, we can identify
  an $\Omega(\ell/n)$-fraction of the remaining $\ell$-big vertices in~$I$.
  This propagation requires $O(m\ell)$ time and only linear space~\cite[Section~5.3]{CohenM94}.
  
  Since we do not know which samples $S_i$ contain a unique big vertex
  and how big the vertices in $I$ are, we run the above (Lemma~\ref{l:locate-ext} plus
  the propagation) for all samples $S_1,\ldots,S_k$ in parallel, and stop
  if we succeed in determining $\Omega(\ell/n\cdot g)=\Omega(\ell/n\cdot |I|)$ new big values
  within $\Ot(m\ell)$ time.
  One can show that such an event occurs with high probability.
  
  Intuitively, since identifying an $\Omega(\ell/n)$ fraction of the remaining
  big values costs $\Ot(m\ell)$ time, the algorithm finishes in $\Ot(mn)$ time.
  The space usage is linear by Lemma~\ref{l:locate-ext} and~\cite[Section~5.3]{CohenM94}.
\end{proof}
Unfortunately, it is not entirely clear whether one can extend the already complicated
algorithm behind Theorem~\ref{t:cohen} so that it produces
a certificate of infeasibility if the system is infeasible.
As a result, our algorithm in this
section does not produce a certificate, but merely computes $\xmax{}$ assuming
the system is feasible, and eventually
possibly declares the system infeasible by checking whether $\xmax{}$ is feasible
(recall, however, that such a guarantee is sufficient e.g., for DMDPs, where feasibility can be taken for granted).
Given the above, we do not require the algorithm presented in this
section to satisfy the additional requirements stated in Remark~\ref{r:prob},
and in the following, we can assume that the system is feasible.

Let $h$ be an integer to be chosen later.
Suppose we stop the phases of the algorithm of Section~\ref{s:simple} after the earliest
phase $j$ with $2^j\geq h$.
Recall that this guarantees, with constant probability,
that the sampled vertices hit all the cycles $C\in\cycs^*$ with $|C|\leq 2^{j+1}$.
So for each $C\in\cycs^*$ with $|C|\leq 2^{j+1}$, 
there exists $v\in V(C)$ such that $\cb(D_v)\leq x^*_v\leq \cb_{v,|C|}$ for
some $D_v\in\cycs_v^n$ satisfying $\gamma(D_v)<1$.
It can be easily verified that the phases $0,\ldots,j$ use
$\Ot(mnh)$ time and linear space.

Generally speaking, our strategy is to reduce handling cycles of $\cycs^*$ of size more than $h$ 
to computing a pointwise maximal solution of a smaller but denser instance of the monotone 2VPI system
feasibility that will be passed
to the algorithm of~\cite{CohenM94}.

Let $\Psi$ denote the set of those cycles $C\in \cycs^*$ such that $|C|>h$.
Let $S$ be a random subset of $V$ of size $\Theta((n/h)\log{n})$.
With high probability (dependent on the constant hidden in the $\Theta$ notation),
$S$ contains not only a vertex of every cycle in $\Psi$, but
also a vertex of \emph{every} subpath of length $h-1$ of any cycle in $\Psi$.
Denote by $S^*$ the set of those $s\in S$ that indeed lie on some cycle in $\Psi$.
Observe that we might have $S^*\subsetneq S$.

We now construct a new, smaller instance $G'=(S,E')$ that will satisfy the following:
\begin{itemize}
  \item If $G$ is feasible, $G'$ is also feasible.
  \item For each $s\in S$, the pointwise maximum $(\xmax{s})'$ of $G'$
    satisfies either $(\xmax{s})'=\infty$ or \linebreak $(\xmax{s})'=c(F)+\gamma(F)\cdot \cb(D)$
    for some $F\in\paths_{s,w}^n$,
    $D\in\cycs_{w}^n$, $\gamma(D)<1$ (where $\paths_{s,w}^n$ and $\cycs_{w}^n$ are sets
    of subgraphs of $G$).
Moreover, if $s\in S^*$, then additionally $(\xmax{s})'\leq \cb_{s,n}$.
\end{itemize}
By the above conditions, setting $x^*_s:=\min(x^*_s,(\xmax{s})')$ after the initial $j+1$ phases
will meet the 
requirements of Problem~\ref{prob}.

For every pair $s,t\in S$, $E'$ will contain a \emph{single} edge $st$
with $(c'(st),\gamma'(st))=(c(Q_{st}),\gamma(Q_{st}))$ where~$Q$ is \emph{any} walk minimizing $c(Q_{st})+\gamma(Q_{st})\cdot \xmax{t}$ among $s\to t$
walks of at most $h$ edges.
Note that $G'$ has $|S|^2$ edges.
Using parametric search~\cite{Megiddo83a}, such walks can be found without computing values $\xmax{t}$ exactly.
Now let us argue that solving the problem on $G'$ is enough for our needs.

\begin{lemma}\label{l:compress}
  For any $s\in S$, either $(\xmax{s})'=\infty$ or $(\xmax{s})'=c(F)+\gamma(F)\cdot\cb(D)$ for some $F\in\paths_{s,w}^n$,
    $D\in\cycs_{w}^n$, $\gamma(D)<1$.
Moreover, with high probability, for each $s\in S^*$ we have
  $(\xmax{s})'\leq \cb_{s,n}$.
\end{lemma}
\begin{proof}
  The first part follows since each inequality in $G'$ corresponds to
  to some chain of inequalities in $G$. Hence, $G'$ is feasible
  and the existence of $(F,D)$ if $(\xmax{s})'\neq\infty$ follows by Theorem~\ref{t:shostak}.
  
  We will actually prove $(\xmax{s})'=\xmax{s}$ for $s\in S^*$. 
  This will imply $(\xmax{s})'\leq \cb_{s,n}$ by Theorem~\ref{t:shostak}.

  Recall that $s\in S^*$ lies on some cycle in $\Psi$.
  Let $s=s_1,s_2,\ldots,s_q$ be all the vertices of $S^*\cap V(C_w)$ in the order
  they appear on $C_w$ starting at $s$.
  Put $s_{q+1}:=s_1$. 
  By Theorem~\ref{t:shostak}, for all $s_i$ we have $\xmax{s_i}=\cb_{s_i}(C_w)$.
  Observe that with high probability, for any $s_i$, $i=1,\ldots,q$, $s_i$ and $s_{i+1}$ are at most~$h$ edges apart,
  as otherwise there would exist a subpath of $C_w$ consisting of $h-1$ edges
  that is not hit by $S$. Let $Q_i$ be the $s_i\to s_{i+1}$ subpath of $C_w$; we have $|Q_i|\leq h$.
  Note that by~\eqref{eq:source-change}, $\xmax{s_i}=\cb_{s_i}(C_w)=c(Q_i)+\gamma(Q_i)\cdot \cb_{s_{i+1}}(C_w)=c(Q_i)+\gamma(Q_i)\cdot \xmax{s_{i+1}}$ holds for any $i$
  and moreover, for any $s_i\to s_{i+1}$ walk $Q_i'\in\bigcup_{k=1}^\infty\paths_{s_i,s_{i+1}^k}$ in $G$
  we have $\xmax{s_i}\leq c(Q_i')+\gamma(Q_i')\cdot \xmax{s_{i+1}}$.
  So $Q_i$ is a path minimizing $c(Q)+\gamma(Q)\cdot \xmax{s_{i+1}}$ over~\emph{all} $s_i\to s_{i+1}$ walks~$Q$ in $G$.
  By $|Q_i|\leq h$ and the definition of $G'$, for each edge $e_i=s_is_{i+1}$ in $G'$ we have
  \begin{equation*}
    c'(e_i)+\gamma'(e_i)\cdot \xmax{s_{i+1}}= c(Q_i)+\gamma(Q_i)\cdot \xmax{s_{i+1}}
  \end{equation*}
  We now argue that the cycle $s_1s_2\ldots s_q$ in $G'$ implies the same upper bound
  on $x_s\equiv x_{s_1}$ that $C$ implies on $x_s$ in $G$.
  This will imply $\xmax{s}\geq (\xmax{s})'$ as needed.
  
  Since $\xmax{s_i}=c(Q_i)+\gamma(Q_i)\cdot \xmax{s_{i+1}}$ for all $i$,
  we also have $\xmax{s_i}=c'(e_i)+\gamma'(e_i)\cdot \xmax{s_{i+1}}$ for all $i$.
  So, by chaining these equations we get:
  \begin{align*}
    \xmax{s_1}&=c(Q_1\cdots Q_q)+\gamma(Q_1\cdots Q_q)\cdot \xmax{s_1},\\
    \xmax{s_1}&=c(e_1\cdots e_q)+\gamma(e_1\cdots e_q)\cdot \xmax{s_1}.
  \end{align*}
  As a result:
  \begin{equation*}
    \frac{c(e_1\cdots e_q)}{1-\gamma(e_1\cdots e_q)}=\xmax{s_1}=\frac{c(Q_1\cdots Q_q)}{1-\gamma(Q_1\cdots Q_q)},
  \end{equation*}
  so indeed the cycle $e_1\ldots e_q$ in $G'$ implies the same upper bound $\xmax{s}$ on $x_s$ as
  $C$ in $G$.
\end{proof}

After constructing $G'$, we will invoke the algorithm of~Theorem~\ref{t:cohen} on it.
Since $G'$ has $\Ot(n/h)$ vertices and $\Ot((n/h)^2)$ edges,
this takes $\Ot((n/h)^3)$ expected time with high probability
and requires $\Ot(n/h)$ space.
The final step is to show how to compute $G'$.

\begin{lemma}\label{l:dense}
  Given $S$, the instance $G'$ can be computed in $\Ot(mnh)$ time and $O(n^2\log{n}/h)$ space.
\end{lemma}
\begin{proof}
  Let us first start with a simpler algorithm computing the edges $st\in E'$
  for a \emph{single vertex} $t\in S$ in $\Ot(mnh)$ time and linear space.
  We apply parametric search analogously as we do in the proof of Lemma~\ref{l:parametric}.\footnote{For a more thorough description
  of an application of parametric search, see the proof of Lemma~\ref{l:parametric}.}
  However, this application is slightly simpler: here we only care about
  the tightest bounds imposed by walks, and not tightest absolute bounds imposed by closed walks,
  so in particular we do not have
  to optimize the values $\gamma(\cdot)$ of the computed walks.
  For each bound on walk length $j=0,\ldots,h$, and $v\in V$, we maintain
  a walk~$P_{j,v}$ (actually, only the values $c(P_{j,v})$ and $\gamma(P_{j,v})$),
  such that $P_{j,v}$ minimizes $c(P)+\gamma(P)\cdot \xmax{t}$
  over all walks $P\in\paths_{v,t}^j$.
  
  To enable computing
  the walk $P_{j,v}$ based on $\{P_{j-1,w},w\in V\}$ without
  knowing $\xmax{t}$, we maintain an interval $(a,b)$ such that $\xmax{t}\in (a,b)$.
  In the first step, we compute the lower envelope of the set $L_v$ of lines
  $y=(c(e)+\gamma(e)\cdot c(P_{j-1,w}))+\gamma(e)\gamma(P_{j-1,w})\cdot x$,
  where $vw=e\in E$,
  and record the breakpoints as $X_v$.
  Computing all $X_v$ takes $O(m\log{m})$ time in total.
  We then merge the sets $X_v$ into a single sorted sequence $X$ in $O(m\log{m})$ time.
  Next, binary search is performed with the aim of
  locating two consecutive elements $x_i,x_{i+1}\in (X\cap (a,b))\cup \{a,b\}$
  such that $\xmax{t}\in (x_i,x_{i+1})$.
  Each decision made by binary search is guided using Lemma~\ref{l:locate-ext}.
  So finding $x_i,x_{i+1}$ takes $O(nm\log{m})$ time.
  Note that for each $v$ there exists a line in $L_v$ that
  attains $\min\{f(x):f\in L_v\}$ for all $x\in (x_i,x_{i+1})$.
  That line can be identified easily in $O(\deg_G(v))$ time
  and the corresponding walk is chosen to constitute~$P_{j,v}$.
  So picking the walks $P_{j,v}$ takes $O(m)$ time.

  Finally, for each $s\in V$ set $c'(st)=c(P_{h,s})$ and $\gamma'(st)=\gamma(P_{h,s})$.
  Note that this satisfies the definition of $G'$.

  Since handling a single walk length bound $j$ takes $O(mn\log{m})$ time,
  constructing the incoming edges of~$t$ in~$G'$ takes $O(mnh\log{m})$ time.
  Clearly, we use $O(m)$ space: when proceeding to the next $j$ we require
  two numbers per walk $P_{j-1,w}$ for each $w\in V$ be stored.
  Also, $O(m)$ space is needed to construct the merged sequence~$X$ of breakpoints.

  To obtain a more efficient solution, albeit at the cost of using randomization,
  we make use of Lemma~\ref{l:cohen}.
  We run the above algorithm in parallel for all $t\in V$.
  For each $j$, we synchronize the $O(\log{m})$ decisions made by
  binary search so that the $i$-th decision is made for all targets $t\in S$
  at the same time. Since each decision is of the form $\xmax{t}<\xi_t$,
  all of them can be made at once in $\Ot(nm)$ expected time using Lemma~\ref{l:cohen}.
  Hence, the running time is $O(|S|mh\log{m})+\Ot(nmh)=\Ot(nmh)$.
  
  Unfortunately, simply running the simple-minded algorithm
  in parallel leads to $O(|S|m)=\Ot(mn/h)$ space usage.
  To obtain an improved $\Ot(n^2/h)$ space bound we need
  to avoid storing the envelopes and the breakpoint sequences $X$
  while proceeding in parallel.
  The solution is to instead maintain for each $t$ only a single
  interval $(a_t,b_t)$ such that $\xmax{t}\in (a_t,b_t)$.
  Before the synchronized decision made by the algorithm of Lemma~\ref{l:cohen},
  for each $t$ we recompute the envelopes and the sequence $X$ from scratch,
  and find the middle element $\xi_t$ of $X\cap (a_t,b_t)$.
  Afterwards, we discard $X$ and the envelopes and only store $\xi_t$.
  In other words, the space used by the envelopes and breakpoints
  is shared between different $t\in S$.
  Since for a single $j$ binary search makes $O(\log{m})$ steps, this adjustment
  increases the $O(|S|mh\log{m})$ term in the running time to $O(|S|mh\log^2{m})=\Ot(nm)$.

  It is however not clear how to eliminate the space cost incurred
  by maintaining the walks $P_{j,w}$ for all targets $t\in S$ at once.
  As a result, the space usage is $\Ot(|S|\cdot n)=\Ot(n^2/h)$.
\end{proof}

\tvpi*
\begin{proof}
  Recall that running the phases $0,\ldots,j$ of the algorithm of Section~\ref{s:simple} until $2^j\geq h$
  takes $\Ot(nmh)$ time and linear space.
  Constructing the dense graph $G'$ takes $\Ot(nmh)$ time and \linebreak $\Ot(n^2/h+m)$ space by Lemma~\ref{l:dense}.
  Finally, solving the instance $G'$ using Theorem~\ref{t:cohen} takes $\Ot(|S|^3)=\Ot((n/h)^3)$ time
  and $\Ot((n/h)^2)$ space.
  As we have argued, these steps are enough to solve Problem~\ref{prob}.
  By Lemma~\ref{l:reduction}, additional computation in $O(nm)$ time and linear space can
  compute a pointwise maximal solution to the input monotone 2VPI system with high probability.

  Recall from the proof of Lemma~\ref{l:simple} that the solution can be verified in $O(nm)$ time
  so that we can make the algorithm Las Vegas by running it until it computes the correct answer.
  With high probability $O(\log{n})$ runs are enough.
\end{proof}
\specific*
\begin{proof}
  For the former item, pick $h=\max(1,n^{1/2}/m^{1/4})$. For the latter, let the $\Ot(n^2/h)$ space
  bound in Theorem~\ref{t:main} be actually $O(n^2/h\cdot \log^c{n})$. If $m=O(n\log^c{n})$,
  just use the simple algorithm of Section~\ref{s:simple}.
  Otherwise, pick $h=\max(1,(n^2\log^c{n})/m)$. Then $h=\Omega(1)$ and $h=O(n)$.
\end{proof}

\begin{remark}
  If the algorithm of Theorem~\ref{t:cohen} could produce a certificate of infeasibility (which is likely, after a suitable extension),
  and the algorithm of Lemma~\ref{l:cohen} could locate values wrt. $x^{\leq}$ without
  assuming feasibility, then the trade-off algorithm could also produce a certificate
  of infeasibility.

  In other words, the ability of the trade-off algorithm
  to certify infeasibility depends solely on whether the results from~\cite{CohenM94}
  can be extended to certify infeasibility.
\end{remark}
\section{Discounted All-Pairs Shortest Paths}
Let $\gamma\in (0,1)$ be a fixed discount factor.
Let $G=(V,E)$ be a weighted directed graph with edge costs given by $c:E\to \mathbb{R}$.
For a walk $P\subseteq G$, its cost $c(P)$ is defined
as it was before but with $\gamma(e)\equiv\gamma$ for all $e\in E$.
So, if $P=e_1\ldots e_k$, then $c(P)=\sum_{i=1}^k c(e_i)\cdot \gamma^{i-1}$.
Since $\gamma(e)=\gamma$ for all $e$, we can assume that $G$ is simple (and thus $m=O(n^2)$)
and thus writing $c(uv)$ is unambiguous.

A \emph{discounted shortest path} from $s$ to $t$ is a finite $s\to t$
walk $P$ minimizing $c(P)$. 
Let $\dist^k_G(s,t)$ be the minimum cost of a $s\to t$
walk with \emph{exactly} $k$ edges.
The \emph{discounted distance} $\dist_G(s,t)$ is defined as $\inf_{k\geq 0}\dist^k_G(s,t)$.
Madani~et~al.~\cite{MadaniTZ10} showed that $\dist_G(s,t)$ is not necessarily attained
by a finite path even if all edge costs are positive. So, in particular, a discounted
shortest path needs not be simple. However, they showed~\cite[Lemma~5.8]{MadaniTZ10} the following reduction.
\begin{lemma}\label{l:apspred}\textup{\cite{MadaniTZ10}}
  In $O(nm)$ time, the discounted APSP problem on $G$ can be reduced to a discounted APSP
  problem on a graph $G'$ with $|V(G')|=O(n)$ and $|E(G')|=O(n+m)$ such that:
  \begin{enumerate}[label=(\arabic*)]
    \item each $v\in V(G)$ is assigned two corresponding vertices $v',v''$ in $V(G')$,
    \item for all $u,v\in V(G)$ we have $\dist_G(u,v)=\dist_{G'}(u',v'')$,
    \item for all $s,t\in V(G')$ there exists a finite simple path $P=s\to t$ in $G'$ such that $c(P)=\dist_{G'}(s,t)$,
    \item if $k$ is minimum such that $\dist_{G'}^k(s,t)=\dist_{G'}(s,t)$, then for
      \emph{every} walk $Q=s\to t$ in $G'$ such that $c(Q)=\dist_{G'}(s,t)$ and $|Q|=k$, $Q$ is a simple path.
  \end{enumerate}
\end{lemma}
\begin{proof}
  The properties (1)-(3) are proved in~\cite[Lemma~5.8]{MadaniTZ10}. We now argue that 
  property~(4) holds as well. Suppose $Q$ is not a simple path, then, it
  can be expressed as $Q_1CQ_2$, where $C$ is a non-empty closed walk.
  We have $c(CQ_2)<c(Q_2)$, as otherwise we would have $\dist_{G'}(s,t)\leq c(Q_1Q_2)\leq c(Q_1CQ_2)=\dist_{G'}(s,t)$
  and $Q_1Q_2$ would have less than $k$ edges, which is impossible.
  Consequently,
  we obtain:
  \begin{equation*}
    \dist_G'(s,t)\leq c(Q_1CCQ_2)=c(Q_1C)+\gamma(Q_1C)\cdot c(CQ_2)<c(Q_1C)+\gamma(Q_1C)\cdot c(Q_2)=c(Q_1CQ_2),
  \end{equation*}
  that is, $\dist_G'(s,t)<\dist_G'(s,t)$,
  a contradiction.
\end{proof}
In the following, we assume that the above reduction has been applied and we reset $G:=G'$. As a result, we only
need to compute $\min_{k=0}^n\{\dist_G^k(s,t)\}$ for all $s,t\in V$.
For simplicity, in the following we put
$\dist_G(s,t):=\min_{k=0}^n\{\dist_G^k(s,t)\}$.
Moreover, for each pair $s,t\in V$, let $\ell_{s,t}$
denote the minimal number of edges on an $s\to t$ walk
with the minimum cost $\dist_G(s,t)$.
By Lemma~\ref{l:apspred}, all discounted shortest $s\to t$ paths
with $\ell_{s,t}$ edges are simple.

We concentrate
on computing the values $\dist_G(s,t)$ and not the respective paths themselves.
However, the given algorithms can be easily extended to produce the representation
of all discounted shortest paths within the same asymptotic time bound.
Note that due to the following simple facts (observed in~\cite{MadaniTZ10}), the discounted
shortest paths to a single fixed target vertex $t\in V$ (from all possible sources) 
can be represented using a tree. 

\begin{observation}\label{o:suffix}
  Let $P=s\to t$ be a discounted shortest path. Then, every suffix of $P$ is
  a discounted shortest path. More formally, if $P=P_1P_2$, where $P_2$ is
  a $w\to t$ path, then $P_2$ is a discounted shortest path, i.e., $c(P_2)=\dist_G(w,t)$.
\end{observation}
\begin{proof}
  We have $c(P)=c(P_1)+\gamma^{|P_1|}\cdot c(P_2)$. If there existed another path
  $Q=w\to t$ with \linebreak $c(Q)<c(P_2)$, $P_1Q$ would have smaller cost than
  $P_1P_2$ since $\gamma>0$. 
  Then $P$ would not be shortest.
\end{proof}

\begin{observation}\label{o:prefix}
  Let $P=s\to t$ be a discounted shortest path. Let $1\leq k\leq |P|$.
  Then the $k$-prefix of $P$ forms a path with the smallest cost
  among paths with exactly $k$ edges and same respective endpoints.
  More formally, if $P=P_1P_2$, where $P_1$ is
  a $s\to w$ path and $|P_1|=k$, then $c(P_1)=\min\{c(P):P\in \paths_{s,w}^k, |P|=k\}$.
\end{observation}
\begin{proof}
  We have $c(P)=c(P_1)+\gamma^{k}\cdot c(P_2)$. If there existed another path
  $Q=s\to w$ with $|Q|=k$, then $\gamma(P_1)=\gamma(Q)$ and $c(Q)<c(P_1)$, so $QP_2$ would have a smaller cost than
  $P_1P_2$.
\end{proof}

In the following lemmas, we introduce the building blocks of our algorithm
for discounted APSP.

\begin{lemma}\label{l:target}
  Let $t\in V$ and $k\geq 1$. One can compute $\dist^{\leq k}_G(s,t)=\min_{j=0}^k\{\dist^j_G(s,t)\}$
  for all $s\in V$ and $j=0,\ldots,k$ in $O(mk)$ time.
\end{lemma}
\begin{proof}
  For $j=0$ we have $\dist^{\leq 0}_G(t,t)=0$ and $\dist^{\leq 0}_G(s,t)=\infty$ for $s\neq t$.
  Note that by Observation~\ref{o:suffix}, for any $j>0$ and $s\in V$
  we have
  \begin{equation*}
    \dist^{\leq j}_G(s,t)=\min\left(\dist^{\leq j-1}_G(s,t),\min_{su\in E}\{c(su)+\gamma\cdot \dist_G^{\leq j-1}(u,t)\}\right).
  \end{equation*}

  As a result, all the values $\dist^{\leq k}_G(s,t)$ of interest can be computed using $k$ Bellman-Ford-style
  steps that take $O(m)$ time each.
\end{proof}

\begin{lemma}\label{l:source}
  Let $s\in V$ and $k\geq 1$. One can compute $\dist^j_G(s,t)$ for all $j=0,\ldots,k$ and $t\in V$
  in $O(mk)$ time. 
\end{lemma}
\begin{proof}
  We have $\dist^{0}_G(t,t)=0$ and $\dist^{0}_G(s,t)=\infty$ for $s\neq t$.
  By Observation~\ref{o:prefix}, for any $j>0$ and $t\in V$
  we have
  \begin{equation*}
    \dist^{j}_G(s,t)=\min_{ut\in E}\{\dist^{j-1}_G(s,u)+\gamma^{j-1}\cdot c(ut)\}.
  \end{equation*}
  As a result, all the required $\dist^j_G(s,t)$ can be computed in $O(mk)$ time.
\end{proof}
\begin{lemma}\label{l:lines}
  Let $s\in V$ and $k\geq 1$. Then, in $O(mk)$ time one can construct
  a data structure $D(s)$ answering the following queries.

  For a given vector $(v,x_1,\ldots,x_l)\in V\times\mathbb{R}^l$ such
  that $x_1\leq \ldots \leq x_l$, compute the values:
  \begin{equation*}
    D(s,v,x_l)=\min_{i=0}^k \{\dist_G^i(s,v)+\gamma^i\cdot x_l\},
  \end{equation*}
  for $j=1,\ldots,l$. The query is answered in $O(k+l)$ time.
\end{lemma}
\begin{proof}
  The first step is to compute $\dist_G^j(s,t)$ for all $j=0,\ldots,k$ and $t\in V$ using Lemma~\ref{l:source}.
  For a given $v$, consider the lines $y=\dist^i_G(s,v)+\gamma^i\cdot x$ for $i=0,\ldots,k$.
  The lower envelope $L_v$ of these lines can be computed in $O(n)$ time
  since the slopes of these lines are sorted in decreasing order.
  Note that $L_v$ can be interpreted as a plot of the function $D(s,v,x)$.
  The data structure consists of the envelopes $L_v$ for all $v\in V$.

  When a query $(v,x_1,\ldots,x_l)$ comes, we match the coordinates $x_1,\ldots,x_l$ to the
  segments of~$L_v$ corresponding to the lines minimizing each subsequent $x_i$.
  Since the coordinates $x_i$ come in increasing order and $|L_v|=O(k)$, answering a query takes $O(k+l)$ time.
\end{proof}

\begin{lemma}\label{l:hitting}
  Let $S\subseteq V$ and $k\geq 1$. Then, in $O(|S|mk)$ time one can find a set $X\subseteq V$
  of size $O(n/k\log{n})$ such that for all $(s,t)\in S\times V$ satisfying $\ell_{s,t}\geq k$,
  there exists a path $Q_{s,t}=s\to t$ with $c(Q_{s,t})=\dist_G(s,t)$
  such that one of the first $k+1$ vertices of $Q_{s,t}$ is contained in $X$.
  We call~$X$ a \emph{$k$-hitting set} of discounted shortest paths from $S$ to~$V$.
\end{lemma}
\begin{proof}
  We first compute $\dist_G^k(s,t)$ for all $(s,t)\in S\times V$ by applying Lemma~\ref{l:source} for
  each $s\in S$. This takes $O(|S|mk)$ time.
  Note that the algorithm in Lemma~\ref{l:source} can be easily extended
  to also return for each $t$ some (not necessarily simple) walk $P_{s,t}$ with $|P_{s,t}|=k$
  and $c(P_{s,t})=\dist^k_G(s,t)$,
  if one exists.
  Now, let $\mathcal{Q}$ denote those computed $P_{s,t}$ \emph{that constitute simple paths},
  i.e., satisfy $V(P_{s,t})=k+1$.

  We now prove that a set $X$ hitting all paths in $\mathcal{Q}$ (i.e., such that $V(P)\cap X\neq\emptyset$ for all $P\in\mathcal{Q}$) is
  a $k$-hitting set of discounted shortest paths from $S$ to $V$.
  Indeed, consider some pair $s,t\in V$ and some path $R$ with $c(R)=\dist_G(s,t)$ and $|R|=\ell_{s,t}\geq k$.
  Let $R=R_1R_2$, where $R_1=s\to w$ and $|R_1|=k$.
  Note that $P_{s,w}$ exists since $R_1$ does. By Observation~\ref{o:prefix}, we have $c(R_1)=\dist^k_G(s,w)=c(P_{s,w})$.
  So $c(P_{s,w}R_2)=\dist_G(s,t)$ and $P_{s,w}R_2$ is a discounted shortest $s\to t$ path.
  If $P_{s,w}$ was not simple, then $P_{s,w}R_2$ would be a non-simple shortest $s\to t$ path
  with the minimum possible number of edges~$\ell_{s,t}$,
  which would contradict the property~(4) guaranteed by Lemma~\ref{l:apspred}.
  So $P_{s,w}$ has to be simple and therefore $P_{s,w}\in\mathcal{Q}$.
  As a result, $V(P_{s,t})\cap X\neq\emptyset$.
  Note that the discounted shortest $s\to t$ path $P_{s,w}R_2$ contains a vertex of $X$ among its first $k+1$ vertices, as desired.

  Since every element of $\mathcal{Q}$ has precisely $k+1$ \emph{distinct} vertices of $V$,
  one can deterministically compute a set $X$ of
  size $O((n/k)\log{n})$ that hits
  all elements in $\mathcal{Q}$ using a folklore greedy algorithm
  that repeatedly includes in $X$ an element that hits the most (which is always at least an \linebreak $\Omega(k/n)$-fraction)
  of the remaining elements of $\mathcal{Q}$. 
  The greedy algorithm can be easily implemented in $O(|\mathcal{Q}|k)=O(|S|nk)$ time.
  For analysis, see e.g.~\cite[Lemma~5.2]{King99}.
\end{proof}

Let us denote denote by $T(r)$ the time needed to compute
discounted shortest paths from~$r$ sources ($1\leq r\leq n$) in $G$,
as a function of $n$ and $m$.
Of course, our ultimate goal is to compute discounted shortest
paths from all~$n$ sources.
We will show two different reductions to a problem with
a smaller number of sources.

\begin{lemma}\label{l:red1}
  There exists a constant $c\geq 1$ such that for
  any $r,h\in [1,n]$ we have:
  \begin{enumerate}[label=(\arabic*)]
    \item $T(r)\leq T(\min(cn/h\cdot \log{n},n))+O(nmh)$,
    \item $T(r)\leq T(\min(cn/h\cdot \log{n},n))+O(rmh+(n^2/h)\log{n}\cdot (r+\log{n}))$.
  \end{enumerate}
\end{lemma}
\begin{proof}
  Let $S$ denote the set of $r$ source vertices.
  Clearly, we can compute discounted shortest paths that contain less than $h$
  edges in $O(rmh)$ by applying Lemma~\ref{l:source} to every source vertex $s\in S$.
  Hence, we can concentrate computing discounted shortest paths
  for such pairs $(s,t)$ that $\ell_{s,t}\geq h$.

  Let $X$ be a $h$-hitting set of discounted shortest paths from $S$ to $V$
  computed using Lemma~\ref{l:hitting}. 
  The set $X$ has size $O(\min((n/h)\log{n},n))$
  and can be computed in $O(rmh)$ time.
  Let us compute the values $\dist_G(s,t)$ for all $(s,t)\in X\times V$ recursively in
  $T(|X|)=T(O(\min(n/h\log{n},n)))$ time.

  Now consider the item (1).
  Let us define $d_j(s,t)$ for $(s,t)\in V\times V$ and $j=0,\ldots,h$ as follows:
  \begin{equation*}
    d_j(s,t)=\begin{cases} \dist_G(s,t) &\text{ if }j=0\text{ and }s\in X,\\
                            \infty &\text{ if }j=0\text{ and }s\notin X,\\
                           \min\left(d_{j-1}(s,t),\min_{su\in E}\{c(su)+\gamma\cdot d_{j-1}(u,t)\}\right)&\text{ if }j>0. 
    \end{cases}
  \end{equation*}
  Similarly as in Lemma~\ref{l:target}, the values $d_j(s,t)$ for a single $t\in V$ can be computed using
  $h$ Bellman-Ford-style phases in $O(mh)$ time. Through all $t$, computing them takes $O(nmh)$ time.

  We now show that if $\ell_{s,t}\geq h$,
  then $d_h(s,t)=\dist_G(s,t)$. Consequently, we can return \linebreak
  $\dist_G(s,t)=\min(d_h(s,t),\dist_G^{\leq h}(s,t))$
  for all $s,t\in S\times V$.

  Indeed, let $QR=s\to t$ be a discounted shortest path such that $|Q|=k$ and there
  exists $x\in V(Q)\cap X$.
  The existence of such a path $QR$ is guaranteed by Lemma~\ref{l:hitting}.
  Let $Q=Q_1Q_2$, where $Q_1=s\to x$.
  By Observation~\ref{o:suffix}, $Q_2R$ is a discounted shortest $x\to t$ path,
  so \linebreak $c(Q_2R)=\dist_G(x,t)=d_0(x,t)$.
  If $Q_1=e_qe_{q-1}\ldots e_1$, where $e_i=v_iv_{i-1}$ then
  using Observation~\ref{o:suffix}
  it can be easily proved inductively that for $j=0,\ldots,q$ we have $d_j(v_j,t)=c(e_j\ldots e_1Q_2R)$.
  So $d_q(s,t)=c(QR)=\dist_G(s,t)$.
  Since $q\leq h$, $d_h(s,t)\leq d_q(s,t)$, and we conclude that $d_h(s,t)=\dist_G(s,t)$ holds as well.
  This concludes the proof of item~(1).

  Let us move to item~(2). Compute the data structure $D(s)$ of Lemma~\ref{l:lines} with $k=h$ for each $s\in S$.
  For every $x\in X$, sort the previously computed values $\dist_G(x,t)$, where $t\in V$.
  This takes $O(|X|n\log{n})=O((n^2/h)\log^2{n})$ time.
  For each $s\in S$, $x\in X$, $t\in V$ compute the values $D(s,x,\dist_G(x,t))$
  as defined in Lemma~\ref{l:lines}.
  This takes $O(r\cdot (n/h)\log{n}\cdot (h+n))=O(r(n^2/h)\log{n})$ time.
  
  We now prove that if $\ell_{s,t}\geq h$ then $\dist_G(s,t)=\min_{x\in X}\{D(s,x,\dist_G(x,t))\}$.
  Clearly, we have $\dist_G(s,t)\leq \min_{x\in X}\{D(s,x,\dist_G(x,t))\}$
  since each $D(s,x,\dist_G(x,t))$ constitutes the discounted cost of some $s\to t$ path in~$G$.
  Similarly as before, there exists such an $s\to t$ path $Q_1Q_2R$ that
  $c(Q_1Q_2R)=\dist_G(s,t)$, $Q_1=s\to y$, $|Q_1|\leq h$ and $y\in V(Q_1)\cap X$.
  Note that $c(Q_2R)=\dist_G(y,t)$ by Observation~\ref{o:suffix}.
  By Observation~\ref{o:prefix}, $c(Q_1)=\dist_G^{|Q_1|}(s,y)$, and thus 
  \begin{equation*}
    D(s,y,\dist_G(s,t))\leq \dist_G^{|Q_1|}(s,y)+\gamma^{|Q_1|}\cdot \dist_G(y,t)=c(Q_1)+\gamma^{|Q_1|}\cdot c(Q_2R)=c(Q_1Q_2R)=\dist_G(s,t).
  \end{equation*}
  We conclude $\min_{x\in X}\{D(s,x,\dist_G(x,t))\}\leq\dist_G(s,t)$ as desired.
\end{proof}

Our discounted APSP algorithm is as follows.
Denote by $\langle S,V\rangle$ the problem of 
computing $\dist_G(s,t)$ for all $(s,t)\in S\times V$.
So our goal is to solve the problem $\langle V,V\rangle$.

Let $d\geq 2$ be to be fixed later.
In the first step, we apply reduction~(1) of Lemma~\ref{l:red1} with $h=d$.
In $O(nmd)$ time this reduces the problem $\langle V,V\rangle$ to $\langle S_1,V\rangle$
for some set $S_1\subsetneq V$ with $|S_1|=O(\min(n,(n/d)\log{n}))$.

Next, for $i=1,\ldots,g=O(\log{n})$ we do the following. 
If $|S_i|=O(\log{n})$, then we solve the problem $\langle S_i,V\rangle$
in a simple-minded way
by applying Lemma~\ref{l:source} with $k=n$ for all $s\in S_i$
in $O(|S_i|nm)=O(nm\log{n})$ time. We set $g=i$ and terminate.

Otherwise, we have $|S_i|=O(\min(n,(n/d^i)\log{n}))$. We reduce the problem to $\langle S_{i+1},V\rangle$
using reduction~(2) of Lemma~\ref{l:red1}
with $h=d^{i+1}$, where $|S_{i+1}|=O(\min(n,(n/d^{i+1})\log{n}))$.
This takes
\begin{equation*}
  O(|S_i|md^{i+1}+(n^2/d^{i+1})\log{n}\cdot (|S_i|+\log{n}))=O\left(nmd\log{n}+(n^3/d^{2i+1})\log^2{n}\right)
\end{equation*}
time.
Note that $g=O(\log_d n)=O(\log{n})$ since $d\geq 2$.
Summing through all $i$ we have:
\begin{equation*}
  \sum_{i=1}^g O\left(nmd\log{n}+(n^3/d^{2i+1})\log^2{n}\right)=O((nmd+(n/d)^3)\log^2{n}).
\end{equation*}
This dominates the cost of the first reduction step and the final step solving $\langle S_g,V\rangle$.
By picking $d=\max(2,n^{1/2}/m^{1/4})$, we obtain the following theorem.
\discount*

\section{Reconstructing minimizing paths in linear space}\label{s:construct-path}
In this section, we prove the following lemma.
\lconstructpath*
\begin{proof}
  We prove the lemma by induction on $k$. If $k=1$, then it is enough to check all \linebreak $s\to t$ edges.
  So suppose $k\geq 2$, and let $k'=\lceil k/2 \rceil$.
  First, let us compute the minimal value \linebreak $(\beta,\Gamma)=(c(Q)+\gamma(Q)\cdot \alpha,\gamma(Q))$.
  This can be done using $k$ propagation steps in $O(mk)$ time and linear space.
  The next step is to compute the pairs $l_v^i$, $r_v^i$, where $v\in V$ and $i=0,\ldots,k'$, defined as:
  \begin{align*}
    l_v^i&=\max_{P\in\paths_{s,v}^i}\left\{\left((\beta-c(P))\cdot \frac{1}{\gamma(P)},\frac{1}{\gamma(P)}\right)\right\},\\
    r_v^i&=\min_{P\in\paths_{v,t}^i}\left\{\left(c(P)+\gamma(P)\cdot \alpha,\gamma(P)\right)\right\},
  \end{align*}
  where minimization/maximization is performed lexicographically.
  Clearly, we have \linebreak $l_v^0=([v=s]\cdot \beta,1)$, and $r_v^0=([v=t]\cdot \alpha,1)$.
  Set $l_v^i=(l_{v,1}^i,l_{v,2}^i)$, and $r_v^i=(r_{v,1}^i,r_{v,2}^i)$.
  For $i>0$, we have the following recursive relations:
  \begin{align*}
    l_v^i&=\max\left(l_v^{i-1},\max_{uv=e\in E}\left\{\left((l_{u,1}^{i-1}-c(e))\cdot \frac{1}{\gamma(e)},\frac{l_{u,2}^{i-1}}{\gamma(e)}\right)\right\}\right),\\
    r_v^i&=\min\left(r_v^{i-1},\min_{vw=e\in E}\left\{\left(c(e)+\gamma(e)\cdot r_{w,1}^{i-1},\gamma(e)\cdot r_{w,2}^{i-1}\right)\right\}\right).
  \end{align*}
  In particular, the values $\{l_v^i,v\in V\}$ depend only on the values $\{l_v^{i-1},v\in V\}$
  and the edges of~$G$, and can be computed in $O(m)$ time based on those.
  The values $\{r_v^i,v\in V\}$ have an analogous property.
  Consequently, we can compute the values $l_v^{k'}$ and $r_v^{k-k'}$ for all $v\in V$ in $O(mk)$ time
  and $O(n)$ space.

  Let us express some optimal path $Q$ as $Q_1Q_2$, where $Q_1\in \paths_{s,z}^{k'}$ and $Q_2\in \paths_{z,t}^{k-k'}$.
  We have
  \begin{equation*}
    c(Q_1)+\gamma(Q_1)\cdot (c(Q_2)+\gamma(Q_2)\cdot \alpha)=c(Q)+\gamma(Q)\cdot \alpha= \beta,
  \end{equation*}
  and hence we obtain:
  \begin{equation*}
    c(Q_2)+\gamma(Q_2)\cdot \alpha= (\beta-c(Q_1))\cdot\frac{1}{\gamma(Q_1)}.
  \end{equation*}
  Note that in fact we have $r_z^{k-k'}=(c(Q_2)+\gamma(Q_2)\cdot \alpha,\gamma(Q_2))$,
  as otherwise we could replace $Q_2$ with an $s\to t$ walk optimizing $r_z^{k-k'}$
  better than $Q_2$ and obtain a better path than $Q$.
  Similarly, $l_{z}^{k'}=\left((\beta-c(Q_1))\cdot\frac{1}{\gamma(Q_1)},\frac{1}{\gamma(Q_1)}\right)$, as otherwise there would exist a walk $Q_1'\in\paths_{s,z}^{k'}$ with
  \begin{equation*}
   \left((\beta-c(Q_1'))\cdot\frac{1}{\gamma(Q_1')},\frac{1}{\gamma(Q_1')}\right) =l_{z}^{k'}> \left((\beta-c(Q_1))\cdot\frac{1}{\gamma(Q_1)},\frac{1}{\gamma(Q_1)}\right).
  \end{equation*}
  Then, we would either have
  \begin{equation*}
    \beta\leq c(Q_1'Q_2)+\gamma(Q_1'Q_2)\cdot \alpha=c(Q_1')+\gamma(Q_1')\cdot  (\beta-c(Q_1))\cdot\frac{1}{\gamma(Q_1)}<c(Q_1')+\gamma(Q_1')\cdot  (\beta-c(Q_1'))\cdot\frac{1}{\gamma(Q_1')},
  \end{equation*}
  which reduces to $\beta<\beta$, i.e., a contradiction, or we would have $c(Q_1'Q_2)+\gamma(Q_1'Q_2)\cdot \alpha=\beta$
  and $\gamma(Q_1'Q_2)<\gamma(Q_1Q_2)$, which would contradict the optimality of $Q=Q_1Q_2$.
  
  It follows that $l^{k'}_{z,1}=r^{k-k'}_{z,1}$ and $\frac{r^{k-k'}_{z,2}}{l^{k'}_{z,2}}=\Gamma$.
  We conclude that in $O(n)$ time we can find some (possibly different than $z$) vertex $x\in V$ such that $l^{k'}_{x,1}=r^{k-k'}_{x,1}$ and $\frac{r^{k-k'}_{x,2}}{l^{k'}_{x,2}}=\Gamma$,
  which is guaranteed to exist by the above.

  By the induction hypothesis,
  one can construct some $Q_2^*\in \paths_{x,t}^{k-k'}$ minimizing the pair\linebreak
  $(c(Q_2^*)+\gamma(Q_2^*)\cdot \alpha,\gamma(Q_2^*))= r_{x}^{k-k'}$
  in $O(m(k-k'))$ time and linear space.
  Similarly, one can construct some $Q_1^*\in\paths_{s,x}^{k'}$ 
  minimizing the pair $\left(c(Q_1^*)+\gamma(Q_1^*)\cdot r_{x,1}^{k-k'},\gamma(Q_1^*)\right)$
  in $O(mk')$ time and linear space.
  Finally, we return the path $Q_1^*Q_2^*$.
  
  We now prove that $c(Q_1^*)+\gamma(Q_1^*)\cdot r_{x,1}^{k-k'}=\beta$ and $\gamma(Q_1^*)\cdot\gamma(Q_2^*)=\Gamma$
  from which the correctness will follow.
  Clearly, $c(Q_1^*)+\gamma(Q_1^*)\cdot r_{x,1}^{k-k'}\geq \beta$.
  Let $Q_1'$ be a path corresponding to $l_x^{k'}$.
  We have:
  \begin{equation*}
    c(Q_1^*)+\gamma(Q_1^*)\cdot r_{x,1}^{k-k'}\leq c(Q_1')+\gamma(Q_1')\cdot r_{x,1}^{k-k'}=c(Q_1')+\gamma(Q_1')\cdot l_{x,1}^{k'}=\beta.
  \end{equation*}
  We conclude that indeed  $c(Q_1^*)+\gamma(Q_1^*)\cdot r_{x,1}^{k-k'}=\beta$.
  It follows that $\gamma(Q_1^*Q_2^*)\geq \Gamma$.
  Suppose $\gamma(Q_1^*Q_2^*)>\Gamma=\frac{\gamma(Q_2^*)}{l_{x,2}^{k'}}=\gamma(Q_1')\cdot \gamma(Q_2^*)$,
  that is, $\gamma(Q_1^*)>\gamma(Q_1')$.
  But then we would have
  \begin{equation*}
    \left(c(Q_1^*)+\gamma(Q_1^*)\cdot r_{x,1}^{k-k'},\gamma(Q_1^*)\right)>\left(c(Q_1')+\gamma(Q_1')\cdot r_{x,1}^{k-k'},\gamma(Q_1')\right),
  \end{equation*}
  which would contradict the definition of $Q_1^*$.
  As a result, $\gamma(Q_1^*Q_2^*)=\Gamma$, as desired.

  The running time $T(n,m,k)$ of the algorithm can be bounded:
  \begin{align*}
    T(n,m,1)&=O(m)\\
    T(n,m,k)&=T(n,m,\lfloor k/2 \rfloor)+T(n,m,\lceil k/2\rceil)+O(mk).
  \end{align*}
  In a standard inductive way one can easily prove $T(n,m,k)=O(mk\log{k})$.
  
  The additional space used is $O(n+k)$, since we only use $\Theta(n)$ space when
  computing the $O(1)$ parameters of recursive calls (of which are $O(k)$ in total).
  After these parameters are found, the $\Theta(n)$ space used for computing
  the values $l_v^i$, $r_v^i$ is not needed anymore.
\end{proof}

\section{Computing the $k$-cycle bounds}\label{s:kcycle}
In this section we consider computing the $k$-cycle bounds $\cb_{v,k}$ as defined
in Section~\ref{s:simplify}.
\begin{lemma}\label{l:locate-cycle}
  Let $v\in V$, $k\leq n$, and $\xi\in \mathbb{R}$.
  There is a linear-space algorithm running in $O(mk)$ time and declaring exactly one of the following statements:
  \begin{enumerate}[label=(\alph*)]
    \item There exists a closed walk $Q\in\cycs_v^k$ with $\gamma(Q)<1$ and $\cb(Q)<\xi$.
    \item There exists a closed walk $Q\in\cycs_v^k$ with $\gamma(Q)>1$ and $\cb(Q)>\xi$.
    \item $\max\{\cb_v(C):C\in\cycs_v^k, \gamma(C)>1\}\leq \xi< \cb_{v,k}$.
    \item There exists a closed walk $Q$ with $c_v(Q)<0$ and $\gamma(Q)=1$ certifying the infeasibility of~\eqref{eq:system}.
    \item $\cb_{v,k}=\xi$.
  \end{enumerate}
  Additionally, after the algorithm finishes and if (a), (b), (d), or (e) is returned, one can request to construct a closed walk $Q$
  certifying the respective statement in $O(mk\log{k})$ time and linear space.

  Note that returning (a) implies $\cb_{v,k}<\xi$, and (c) implies $\cb_{v,k}\geq \xi$.
  In particular, if~\eqref{eq:system} is feasible, then (a) is returned iff $\cb_{v,k}<\xi$,
  whereas $\cb_{v,k}\geq\xi$ iff either (b) or (c) is returned.
\end{lemma}
\begin{proof}
  Propagate the bounds $y_v=\xi$, $y_u=\infty$ for all $u\neq v$, through paths of length at most $k$,
  and also maintain the values $\gamma_w$ such that after the $i$-th propagation
  we have $y_w\neq\infty$ implies that there exists a walk $P\in\paths_{w,v}^i$ with
  $(c(P)+\gamma(P)\cdot\xi=y_w,\gamma(P)=\gamma_w)$ lexicographically minimal.
  This clearly takes $O(mk)$ time and requires linear space.

  Let $Q$ be any (possibly empty) closed walk such that
  $y_v=c(Q)+\gamma(Q)\cdot\xi$ and $\gamma_v=\gamma(Q)$ after the propagation finishes.
  
  Since $\cycs_v$ are actually paths $v\to v$,
  we have $y_v\leq c_v(C)+\gamma(C)\cdot \xi$ for all $C\in\cycs_v^k$ after the propagation.
  As propagation never increases bounds,
  we have $y_v\leq \xi$.

  If eventually $y_v<\xi$ and $\gamma(Q)<1$ both hold, then $c_v(Q)+\gamma(Q)\cdot\xi<\xi$ and we obtain
  $\frac{c_v(Q)}{1-\gamma(Q)}=\cb_v(Q)<\xi$ so indeed $\cb_{v,k}<\xi$.
  Thus, in this case, item (a) with the closed walk $Q$ is returned.

  If $y_v<\xi$ and $\gamma(Q)>1$ then  $\cb_v(Q)=\frac{c_v(Q)}{1-\gamma(Q)}>\xi$.
  We thus return (b) with closed walk $Q$.
  Note that if the system is feasible, for any $C'\in\cycs_v$ with $\gamma(C')<1$
  we need to have $\cb_v(C')\geq \cb_v(Q)$. So we indeed have $\phi_{v,k}\geq\xi$
  if the system is feasible.

  If $y_v<\xi$ and $\gamma(Q)=1$, then we obtain $c_v(Q)<0$.
  Recall that the system is infeasible in such a case and we return item (c)
  with a certificate of infeasibility $Q$.

  Now suppose $y_v= \xi$.
  Then, for all closed walks $C'\in\cycs_v^k$ we have $c(C')+\gamma(C')\cdot\xi\geq \xi$.
  In particular, if $\gamma(C')<1$, then this implies 
  $\frac{c_v(C')}{1-\gamma(C')}\geq \xi$,
  whereas if $\gamma(C')>1$ then  $\frac{c_v(C')}{1-\gamma(C')}\leq \xi$.
  If $\gamma(Q)<1$, then we have $\cb_{v,k}=\cb_v(Q)=\xi$.
  As a result, we can return (e) in this case.
  Otherwise, $\gamma(Q)=1$ which proves that $\cb_{v,k}>\xi$ since
  we minimized $\gamma(Q)$ in the second place.
  We can thus return (c).

  By Lemma~\ref{l:construct-path}, the closed walk can be constructed
  in additional $O(mk\log{k})$ time and linear space if needed.
\end{proof}

\newcommand{\ubar}[1]{\text{\b{$#1$}}}

\lparametric*
\begin{proof}
  First of all, we can check in $O(mk)$ time if $\cb_{v,k}$ is finite by checking whether
  a closed walk $C_{\min}\in \cycs_v^k$ that minimizes $\gamma(C_{\min})$ 
  satisfies $\gamma(C_{\min})<1$. Otherwise, we have $\cb_{v,k}=\infty$.
  Note that some closed walk $C_{\min}$ can be computed in $O(mk\log{k})$ time
  and linear space by applying Lemma~\ref{l:construct-path} with $\alpha=1$
  and all edge costs $c(e)$ for $e\in E$ reset to $0$.
  So in the following assume that $\cb_{v,k}$ is finite.
  
  Analogously, we find a closed walk $C_{\max}\in\cycs_v^k$ maximizing $\gamma(C_{\max})$ by applying
  Lemma~\ref{l:construct-path} with $\alpha=1$ and $c(e):=0$, $\gamma(e):=\frac{1}{\gamma(e)}$ for all $e\in E$.
  If $\gamma(C_{\max})\leq 1$, then there exists no closed walks $C\in \cycs_v^k$ with $\gamma(C)>1$.
  If $\gamma(C_{\max})>1$ and $\cb(C_{\max})>\cb(C_{\min})$, the system is infeasible
  and these two walks constitute a certificate, so we can terminate.
  In the following, let us assume that $\gamma(C_{\max})\leq 1$ or $\cb(C_{\max})\leq \cb(C_{\min})$ holds.

  Roughly speaking, for each $j=1,\ldots,k$ and $w\in V$ we wish to compute a walk $P_{j,w}\in\paths^j_{w,v}$ (actually, only the values
  $c(P_{j,k})$ and $\gamma(P_{j,k})$) such that
  $P_{j,k}$ incurs the tightest upper bound on $x_w$ with respect to $\cb_{v,k}$.
  In other words, $P_{j,w}$ minimizes $c(P_{j,w})+\gamma(P_{j,w})\cdot \cb_{v,k}$ among all $w\to v$ walks in $G$ with at most $j$ edges.

  More specifically, we maintain an interval $a,b\in \mathbb{R}\cup \{-\infty,\infty\}$ (where $a<b$), and possibly
  $C_{\min},C_{\max}\in\cycs^k_v$ such that, after the walks $P_{j,w}$ are computed,
  the following invariants are satisfied.
  \begin{enumerate}
    \item Each $P_{j,w}$ lexicographically minimizes the
  pair $\langle c(P_{j,w})+\gamma(P_{j,w})\cdot z,\gamma(P_{j,w})\rangle$
      \emph{for all}~${z\in (a,b)}$ among all walks in $\paths_{w,v}^j$.
    \item Exactly one of the following holds:
      \begin{enumerate}
        \item $\max\{\cb_v(C):C\in\cycs_v^k, \gamma(C)>1\}\leq a<\cb_{v,k}<b$,
        \item The two walks $C_{\min},C_{\max}\in\cycs_v^k$ satisfy $\gamma(C_{\min})<1$, $\gamma(C_{\max})>1$  and\linebreak
          $a<\cb(C_{\max})\leq \cb(C_{\min})<b$.
      \end{enumerate}
  \end{enumerate}
  Initially $a=-\infty$, $b=\infty$. $P_{0,v}$ is an empty walk,
  and all other walks $P_{0,w}$ for $w\neq v$ do not exist.
  One can easily verify that invariant 1. is satisfied initially,
  Note that invariant 2.(a) is satisfied initially if $\gamma(C_{\max})\leq 1$, and
  2.(b) is satisfied if $\cb(C_{\max})\leq \cb(C_{\min})$.

  The algorithm will only shrink the interval $(a,b)$.
  In particular, if the inequality \linebreak $\max\{\cb_v(C):C\in\cycs_v^k, \gamma(C)>1\}\leq a$ of invariant 2.(a) ever holds for the first time,
  it will hold till the end because $a$ can only increase.

  Suppose we have computed all the walks $P_{j-1,u}$, $u\in V$. Consider some desired walk $P_{j,w}$ that
  we want to compute. If $|P_{j,w}|<j$, then we can set $P_{j,w}=P_{j-1,w}$. Otherwise,
  $|P_{j,w}|\geq 1$. Let $P_{j,w}=eQ$, where $e=wu$.
  Then the lower bound that $P_{j,w}$ incurs on $x_w$ with respect to $\cb_{v,k}$ is:
  \begin{equation*}
    c(eQ)+\gamma(eQ)\cdot \cb_{v,k}=c(e)+\gamma(e)\cdot (c(Q)+\gamma(Q)\cdot \cb_{v,k}).
  \end{equation*}
  Since $\gamma(e)>0$, for a fixed $e$, the above is minimized for $Q=P_{j-1,u}$.
  So $c(P_{j,w})+\gamma(P_{j,w})\cdot \cb_{v,k}$ equals:
  \begin{equation*}
    \min\left(c(P_{j-1,w})+\gamma(P_{j-1,w})\cdot \cb_{v,k},\min_{wu=e\in E}\left\{(c(e)+\gamma(e)\cdot c(P_{j-1,u}))+\gamma(e)\cdot \gamma(P_{j-1,u})\cdot \cb_{v,k})\right\}\right).
  \end{equation*}
  However, the inner minimum above is not easy to evaluate without actually knowing $\cb_{v,k}$.
  But we don't need the minima, but rather the best walk $eP_{j-1,u}$ for each $w\in V$.
  To this end, we first compute, for all $w\in V$ the lower envelope $L_w$
  of the lines 
  \begin{equation*}
    y=\gamma(e)\cdot \gamma(P_{j-1,u})\cdot x + (c(e)+\gamma(e)\cdot c(P_{j-1,u})).
  \end{equation*}
  Using a standard convex-hull-like algorithm this takes $O(\deg_G(w)\log n)$ time.
  We record the $x$-coordinates of the breakpoints of the envelope as $X_w$.
  Through all $w$, computing the envelopes and breakpoints takes $O(m\log{n})$ time.

  Let $X=\{a,b\}\cup \left(\left(\bigcup_{w\in V}X_w\right)\cap (a,b)\right)$.
  We sort the set $X$ in $O(m\log{m})$ time. Then, roughly speaking, we perform a binary search for
  the two neighboring $x_i,x_{i+1}\in X$ such that ${\cb_{v,k}\in (x_i,x_{i+1})}$.
  The binary search makes decisions based on the output of the query $(v,k,\xi=x_i)$
  issued to the algorithm of Lemma~\ref{l:locate-cycle}.
  We now describe how this decision is made in more detail.

  We maintain that $a,b\in X$ and repeatedly 
  shrink the interval $[a,b]$ until $a$ and $b$ are neighboring elements of $X$.
  Let $g$ be the middle element of $X$ between $a$ and $b$.
  We issue the query $(v,k,\xi=g)$ to the procedure of Lemma~\ref{l:locate-cycle}.
  We have several cases depending on the result returned:
  \begin{enumerate}[label=(\alph*)]
    \item We set $b=g$ and reset $C_{\min}$ to the closed walk $C$ returned by Lemma~\ref{l:locate-cycle}.
    \item We set $a=g$ and reset $C_{\max}$ to the closed walk $C$ returned by Lemma~\ref{l:locate-cycle}.
    \item We set $a=g$. 
    \item The system is infeasible and we terminate with the certificate returned
  by Lemma~\ref{l:locate-cycle}.
    \item We have $\cb_{v,k}=g$ and we terminate with the certificate returned by Lemma~\ref{l:locate-cycle}.
  \end{enumerate}

  If, after applying one of the cases (a) or (b) above, we have $\cb(C_{\min})<\cb(C_{\max})$,
  then we terminate with the pair $(C_{\min},C_{\max})$ as the certificate of infeasibility.
  We also stress that in cases (a) and (b), we do not construct $C_{\min}$ or $C_{\max}$ explicitly so that
  the query costs $O(mk)$ instead of $O(mk\log{k})$ time.
  Instead, we only store the returned walk $C$ implicitly by storing $\cb(C)$ and the $O(1)$-sized parameters to the respective
  call of Lemma~\ref{l:locate-cycle}.
  If we later actually need $C_{\min}=C$ or $C_{\max}=C$ when returning a certificate of infeasibility, we call
  Lemma~\ref{l:locate-cycle} once again to reconstruct the closed walk $C$ in $O(mk\log{k})$ time.

  We now argue that while binary search proceeds without terminating prematurely
  with either $\cb_{v,k}$ or a certificate, invariants $1$. and~$2$. are satisfied.
  The cases (d) and (e) terminate the algorithm and thus require no explanation.
  Also, suppose we have $\cb(C_{\max})\leq \cb(C_{\min})$.
  
  That invariant 1. holds follows by the construction.
  Consider invariant 2 and first suppose invariant 2.(a) is satisfied
  before the binary search step.
  Then, clearly cases (a) and (b) maintain that 2.(a) is satisfied.
  However, if case (c) applies, we obtain
  \begin{equation*}
  \max\{\cb_v(C):C\in\cycs_v^k, \gamma(C)>1\}\leq a<\cb_{v,k}
  \end{equation*}
  by Lemma~\ref{l:locate-cycle}.
  But we also have $\cb_{v,k}\leq \cb(C_{\min})<b$, so invariant 2.(b)
  is satisfied afterwards.

  Now suppose that invariant 2.(b) is satisfied before the binary search step.
  For 2.(b) to be satisfied afterwards, it is enough to prove that $a<\cb_{v,k}<b$.
  As $\max\{\cb_v(C):C\in\cycs_v^k, \gamma(C)>1\}\leq a$ before, case (b) cannot actually happen. 
  If case (a) applies, then indeed the returned $C$ certifies that
  $\cb_{v,k}<g$, so we have $a<\cb_{v,k}<b$ afterwards.
  If case (3) applies, then we obtain that $g<\cb_{v,k}$, so indeed $a<\cb_{v,k}<b$ afterwards.

  The binary search clearly issues $O(\log{m})$ queries to Lemma~\ref{l:locate-cycle},
  which takes $O(mk\log{m})$ time for a single $j$.

  After the binary search finishes, $a,b$ are some neighboring
  elements $x_i,x_{i+1}$ of $X$.
  For each
  $w\in V$, we pick the walk $P_{j,w}=eP_{j-1,u}$ such that it
  attains the minima in the interval $(x_i,x_{i+1})$.
  If there are ties, we pick such $eP_{j-1,u}$ that
  $\gamma(e)\gamma(P_{j-1,u})$ is minimized.
  This is well-defined since by the definition of breakpoints $x_i,x_{i+1}$, 
  for any two walks $R,T$ of the form $eP_{j-1,u}$
  we have either $c(R)+\gamma(R)\cdot z\leq c(T)+\gamma(T)\cdot z$ for all
  $z\in (x_i,x_{i+1})$, or 
  $c(R)+\gamma(R)\cdot z\geq c(T)+\gamma(T)\cdot z$ for all
  $z\in ((x_i,x_{i+1})$.
  Picking the right walk $eP_{j-1,u}$ takes $O(m)$ time.
  
  Consider the walk $P_{k,v}$. Recall that $P_{k,v}$ minimizes
  $c(P_{k,v})+\gamma(P_{k,v})\cdot z$ for all $z\in (a,b)$.
  Observe that if invariant 2.(b) is satisfied,
  then, since $\cb(C_{\min}),\cb(C_{\max})\in (a,b)$, we obtain:
  \begin{equation}\label{eq:tmp}
    c(P_{k,v})+\gamma(P_{k,v})\cdot  \frac{c(C_{\min})}{1-\gamma(C_{\min})}\leq c(C_{\min})+\gamma(C_{\min})\cdot \frac{c(C_{\min})}{1-\gamma(C_{\min})}=\frac{c(C_{\min})}{1-\gamma(C_{\min})},
  \end{equation}
  \begin{equation}\label{eq:tmp2}
    c(P_{k,v})+\gamma(P_{k,v})\cdot  \frac{c(C_{\max})}{1-\gamma(C_{\max})}\leq c(C_{\max})+\gamma(C_{\max})\cdot \frac{c(C_{\max})}{1-\gamma(C_{\max})}=\frac{c(C_{\max})}{1-\gamma(C_{\max})}.
  \end{equation}
  On the other hand, if invariant 2.(a) is satisfied, then for $C^*\in\cycs_{v,k}$
  such that $\cb(C^*)=\cb_{v,k}\in (a,b)$, we have:
  \begin{equation}\label{eq:tmps}
    c(P_{k,v})+\gamma(P_{k,v})\cdot  \frac{c(C^*)}{1-\gamma(C^*)}\leq c(C^*)+\gamma(C^*)\cdot \frac{c(C^*)}{1-\gamma(C^*)}=\frac{c(C^*)}{1-\gamma(C^*)}.
  \end{equation}
  
  Consider the case $\gamma(P_{k,v})<1$ and invariant 2.(b) is satisfied. Then it follows from~\eqref{eq:tmp2} that $\cb(P_{k,v})\leq \cb(C_{\max})$.
  If $\cb(P_{k,v})< \cb(C_{\max})$, then the pair $(P_{k,v},C_{\max})$ deems the system infeasible.
  Otherwise, $\cb(P_{k,v})=\cb(C_{\max})$. So, $\cb_{v,k}\leq \cb(P_{k,v})$, and if the system is feasible, then $\cb_{v,k}=\cb(P_{k,v})$.
  We issue a query $(v,k,\xi=\cb(P_{k,v}))$ to Lemma~\ref{l:locate-cycle}.
  If the answer to that query is (e), we indeed have $\cb_{v,k}=\cb(P_{k,v})$ regardless of
  whether the system is feasible.
  Note that if the system is feasible, the procedure of Lemma~\ref{l:locate-cycle} will surely return (e).
  On the other hand, if the answer is (a) or (b) then the produced walk $Q$ along with $C_{\max}$
  or $P_{k,v}$ respectively forms a certificate of infeasibility.
  Clearly, if the answer is (e), we also have a relevant certificate.
  Finally, answer (c) cannot arise since $\cb_{v,k}\leq \cb(P_{k,v})$.
  Note that $P_{k,v}$ can be constructed in $O(mk\log{k})$ time and linear space using
  Lemma~\ref{l:construct-path} with $\alpha=\cb(C_{\max})$.

  If $\gamma(P_{k,v})<1$ and invariant 2.(a) is satisfied
  then $\cb_{v,k}\leq \frac{c(P_{k,v})}{1-\gamma(P_{k,v})}=\cb(P_{k,v})$
  by the definition of $\cb_{v,k}$
  and thus indeed $\cb_{v,k}=\frac{c(P_{k,v})}{1-\gamma(P_{k,v})}$ since~\eqref{eq:tmps} can be rewritten as
  \begin{equation*}
    \cb(P_{k,v})=\frac{c(P_{k,v})}{1-\gamma(P_{k,v})}\leq \frac{c_v(C^*)}{1-\gamma(C^*)}=\cb_{v,k}.
  \end{equation*}
  So in this case we return $\cb(P_{k,v})$ as $\cb_{v,k}$.

  In the following let us set $D=C^*$ if invariant 2.(a) is satisfied and $D=C_{\min}$ if invariant 2.(b)
  is satisfied. Note that we have $\cb(D)\in (a,b)$.

  Now assume that $\gamma(P_{k,v})>1$.   We now prove that then the system is infeasible and we show how
  to obtain a relevant certificate.
  Note that since we break ties while picking a walk $P_{k,v}$
  minimizing ${c(P_{k,v})+\gamma(P_{k,v})\cdot z}$
  by minimizing $\gamma(P_{k,v})$, $\gamma(P_{k,v})>1$
  implies that inequality~\eqref{eq:tmp} is actually strict if $D=C_{\min}$,
  and inequality~\eqref{eq:tmps} is strict if $D=C^*$ applies.
  We conclude that $\cb(P_{k,v})>\cb(D)$.
  If invariant 2.(b) is satisfied,
  it follows that indeed the system is infeasible
  and the pair $(P_{k,v},C_{\min})$ is a certificate.
  If invariant 2.(a) is satisfied, we have a contradiction with $\cb(P_{k,v})\leq a<\cb(C^*)$,
  so such a case cannot actually happen.

  Suppose $\gamma(P_{k,v})=1$. By~\eqref{eq:tmp}~or~\eqref{eq:tmps} (whichever applies), we obtain $c(P_{k,v})\leq 0$. 
  If $c(P_{k,v})<0$, then $P_{k,v}$ is a certificate of infeasibility.
  We now prove $c(P_{k,v})=0$ is actually impossible.
  If this was the case, we would have equality in~\eqref{eq:tmp}~or~\eqref{eq:tmps}, which would imply
  that the lines $y=x$ and $y=c(D)+\gamma(D)\cdot x$ cross at $x=\cb(C_{\min})$.
  As a result, we would have $c(D)+\gamma(D)\cdot x<c(P_{k,v})+\gamma(P_{k,v})\cdot x$ for
  all $x\in \left(\cb(D),b\right)$, where $\cb(C_{\min})<b$ by invariant 2.(b).
  However, $P_{k,v}$ minimizes $c(P_{k,v})+\gamma(P_{k,v})\cdot x$ for \emph{all}
  $x\in (a,b)$, in particular for all
  $x\in \left(\cb(D),b\right)$, a contradiction.

  Since there are $O(k)$ different walk lengths $j$, the running time is $O(mk^2\log{m})$.
  For each walk length $j=1,\ldots,k$ we use $O(m)$ space to store the envelopes and the sorted
  breakpoints.
  However, these are discarded when proceeding to the next $j$.
  When moving to the next path length $j$ we only use $O(n)$ space:
  two numbers per each walk $P_{j-1,u}$, and possibly an implicit $O(1)$-space representations of $C_{\min}$, $C_{\max}$.
  Recall that applying Lemma~\ref{l:locate-cycle} requires $O(m+k)$ space.
\end{proof}

\section{Proof of Lemma~\ref{l:locate-ext}}\label{s:locateext}
In this section, we prove the following lemma.
\locateext*

  Although the algorithm is identical to that of~\cite[Section~4]{RestrepoW09}, 
  in order to prove the additional properties, we will need to describe it (using our notation)
  and re-prove some of its known properties.

  Suppose that for each $w\in V$ and $j=1,\ldots,n$ we wish to compute a
  walk $Q_{w,j}\in \bigcup_{v\in V}\paths_{v,w}^j$ such that:
  \begin{itemize}
    \item In the first place, $Q_{w,j}$ maximizes the value $y_{w,j}=\frac{1}{\gamma(Q_{w,j})}(\xi_v-c(Q_{w,j}))$.
    \item If there are ties, $Q_{w,j}$ minimizes its number of edges $|Q_{w,j}|$.
    \item If there are still ties, $Q_{w,j}$ maximizes its ultimate edge $p_{w,j}$. Here, we assume
      that edges are ordered according to some arbitrary but fixed order.
  \end{itemize}
  For $j=0$, we have $y_{w,0}=\xi_w$ for all $w\in V$, and $p_{w,0}=\perp$.
  Observe that for $j=1,\ldots,n$, we have the following recursive relations:
  \begin{equation*}
    y_{w,j}=\max\left(y_{w,j-1}, \max_{uw=e\in E}\left\{\frac{1}{\gamma(e)}(y_{u,j-1}-c(e))\right\}\right),
  \end{equation*}
  \begin{equation*}
    p_{w,j}=\begin{cases}
      p_{w,j-1} & \text{ if }y_{w,j}=y_{w,j-1},\\
      \argmax_{uw=e\in E}\left\{\left(\frac{1}{\gamma(e)}(y_{u,j-1}-c(e)),e\right)\right\} & \text{ if }y_{w,j}<y_{w,j-1}.
      \end{cases}
  \end{equation*}
  As a result, the values $y_{w,j},p_{w,j}$ for $w\in V$ can be computed solely using the values $y_{u,j-1},p_{u,j-1}$ for $u\in V$
  in $O(m)$ time.

  Note that if $p_{v,j}=e=uv$ then we have $y_{v,j}=(y_{u,j-1}-c(e))/\gamma(e)$, and $p_{v,j}=\perp$ iff $y_{v,j}=\xi_v$.

  The algorithm proceeds in phases. When entering the $j$-th phase, $j=0,\ldots,n$, we will
  only require the values $y_{u,j-1},p_{u,j-1}$ for $u\in V$.
  Based on them we will compute the values $y_{w,j},p_{w,j}$ for $w\in V$.
  We will consider the graphs $H_l=(V,\{p_{w,l}:w\in V, p_{w,l}\neq\perp\})$ for $l\in \{j,j-1\}$,
  both of which consist of some (possibly empty) simple cycles with some out-trees attached.
  
  If $H_j$ contains a cycle $C$ with $\gamma(C)<1$, we will terminate the algorithm with a relevant certificate.
  Otherwise, if for all $j=1,\ldots,n$ the graph $H_j$ does not contain such a cycle, we will declare that
  $\xmax{v}\geq \xi_v$ for all $v\in V$.
  In the following, we give the algorithm's details and analysis.
  \begin{lemma}\label{l:rw1}
    Suppose $H_j$ contains a cycle $C$. Then: 
    \begin{itemize}
      \item for some $w\in V(C)$ we have $y_{w,j}>y_{w,j-1}$,
      \item for all $w\in V(C)$ such that $y_{w,j}>y_{w,j-1}$ we have $y_{w,j-1}\cdot (1-\gamma(C))>c_w(C)$,
      \item for all $w\in V(C)$ we have $y_{w,j}\cdot (1-\gamma(C))\geq c_w(C)$.
    \end{itemize}
  \end{lemma}
  \begin{proof}
    Let us first prove that for some $w^*\in V(C)$ we have $y_{w^*,j}>y_{w^*,j-1}$.
    Since $C\not\subseteq H_0$, there has to be some last iteration $i$, $1\leq i\leq j$ when
    $C\not\subseteq H_{i-1}$.
    Since there exists some edge $wu=e\in E(C)\setminus H_{i-1}$, $p_{w,i}\neq p_{w,i-1}$,
    and, as a result, $y_{w,i}>y_{w,i-1}$.
    Let the vertices appearing on $C$ be $w=u_{0},u_1,\ldots,u_{k-1}$, so that
    an edge $e_i=u_iu_{(i+1)\bmod k}$ belongs to $C$ for all $i=0,\ldots,k-1$.

    We have $y_{u_0,i-1}<y_{u_0,i}=(y_{u_{k-1},i-1}-c(e_{k-1}))/\gamma(e_{k-1})$.
    For any $l=1,\ldots,j-i$, by induction on $l$ one obtains $y_{u_l,i+l-1}<y_{u_l,i+l}$:
    \begin{equation*}
      y_{u_l,i+l-1}=(y_{u_{l-1},i+l-2}-c(e_{l-1}))/\gamma(e_{l-1})<(y_{u_{l-1},i+l-1}-c(e_{l-1}))/\gamma(e_{l-1})=y_{u_l,i+l}.
    \end{equation*}
    In particular, for $l=j-i$ we obtain $y_{u_{j-1},j-1}<y_{u_{j-1},j}$.
    By setting $w^*:=u_{j-1}$, we obtain $y_{w^*,j-1}<y_{w^*,j}$, as desired.

    Once again, let us take some $w\in V(C)$ that $y_{w,j}>y_{w,j-1}$. Let the vertices appearing on $C$ be $w=u_{0},u_1,\ldots,u_{k-1}$, so that
    an edge $e_i=u_iu_{(i+1)\bmod k}$ belongs to $C$ for all $i=0,\ldots,k-1$.

    For each $e_i\in E(C)\subseteq E(H_j)$ we have $y_{u_{(i+1)\bmod k},j}=(y_{u_i,j-1}-c(e_i))/\gamma(e_i)\leq (y_{u_i,j}-c(e_i))/\gamma(e_i)$.
    As a result, by chaining these inequalities we obtain
    \begin{equation*}
      y_{w,j}=y_{u_0,j}\leq \frac{y_{u_{1},j}-c(e_1\ldots e_{k-1})}{\gamma(e_1\ldots e_k)}=\frac{\frac{y_{v_0,j-1}-c(e_0)}{\gamma(e_0)}-c(e_1\ldots e_k)}{\gamma(e_1\ldots e_k)}=\frac{y_{w,j-1}-c_w(C)}{\gamma(C)}.
    \end{equation*}
    But $y_{w,j}>y_{w,j-1}$, so we obtain
    \begin{equation*}
      y_{w,j-1}\cdot \gamma(C)<y_{w,j-1}-c_w(C),
    \end{equation*}
    from which the first item of the lemma follows.
    However, observe that if we picked $w$ arbitrarily, we would obtain
    \begin{equation*}
      y_{w,j}\cdot \gamma(C)\leq y_{w,j-1}-c_w(C)\leq y_{w,j}-c_w(C),
    \end{equation*}
    which implies the second item of the lemma.
  \end{proof}
  In particular, Lemma~\ref{l:rw1} implies that $H_j$ does not contain cycles $C$ with $\gamma(C)=1$,
  as otherwise for some $w\in V(C)$ we would have $c_w(C)<0$, which would make the system infeasible.

  \begin{lemma}\label{l:rw2}
    Suppose the algorithm did not terminate before phase $j$ and
    $H_j$ contains a cycle $C$ with $\gamma(C)<1$. Then for some $w\in V(C)$:
    \begin{itemize}
      \item $H_{j-1}$ contains a simple path $P_w=v\to w$ such that $\xmax{v}<\xi_v$,
      \item $c(P_w)+\gamma(P_w)\cdot \cb_w(C)<\xi_v$,
      \item $|V(P_w\cup C)|\geq j/2$.
    \end{itemize}
  \end{lemma}
  \begin{proof}
    First of all, note that $H_{j-1}$ does not contain cycles $C'$ with $\gamma(C')\leq 1$.
    This follows by Lemma~\ref{l:rw1} and the assumption that the algorithm has not terminated before phase $j$.
    
    Let $w\in V(C)$ be such that $y_{w,j}>y_{w,j-1}$, which implies $y_{w,j-1}>\cb_w(C)$. Such a $w$ exists by Lemma~\ref{l:rw1}.
    Let $P_w$ be the subgraph of $H_{j-1}$ obtained by following the reverse walk consisting of
    $p_{w,j-1}=v_1w$, $p_{v_1,j-1}=v_2v_1$, $p_{v_2,j-1}=v_3v_2$, and so on.
    By the structure of $H_{j-1}$, $P_w$ either (i) is a simple path $v\to w$ such that $p_{v,j-1}=\perp$, 
    or (ii) consists of a path $P'=v\to w$ and a cycle $C'=v\to v$ with $\gamma(C')>1$.
    
    We now argue that if the system is feasible, case (ii) cannot in fact arise.
    By Lemma~\ref{l:rw1} we have that $y_{v,j-1}\cdot (1-\gamma(C'))\geq c_w(C')$,
    i.e., $y_{v,j-1}\leq \cb_v(C')$.
    By chaining inequalities corresponding to the edges on $P'$, we have:
    $y_{w,j-1}\leq (y_{v,j-1}-c(P'))/\gamma(P')\leq (\cb_v(C')-c(P'))/\gamma(P')$,
    which is equivalent to
    \begin{equation*}
      c(P')+\gamma(P')\cdot y_{w,j-1}\leq \cb_v(C').
    \end{equation*}
    By applying $y_{w,j-1}>\cb_w(C)$, we have
    \begin{equation*}
      c(P')+\gamma(P')\cdot \cb_w(C) < \cb_v(C').
    \end{equation*}
    The left-hand side above is an upper bound on $x_v$, and
    the right-hand side is a lower bound on $x_v$. Therefore, the system is infeasible, so this case cannot happen.

    Let us thus focus on the case (i). By chaining the inequalities on $P_w$, we obtain
    \begin{equation}\label{eq:rw1}
      c(P_w)+\gamma(P_w)\cdot y_{w,j-1}\leq \xi_v,
    \end{equation}
    and by applying  $y_{w,j-1}>\cb_w(C)$ once again, we have
    \begin{equation*}
      c(P_w)+\gamma(P_w)\cdot \cb_w(C)< \xi_v,
    \end{equation*}
    which proves that a tighter upper bound on $x_v$ than $\xi_v$ is possible, i.e., $\xmax{v}<\xi_v$.
    
    Finally, by combining~\eqref{eq:rw1} and Lemma~\ref{l:rw1}, we have:
    \begin{align*}
      (\xi_v-c(P_wC))\cdot\frac{1}{\gamma(P_wC)}&=(\xi_v-c(P_w)-\gamma(P_w)c_w(C))\frac{1}{\gamma(P_wC)}\\
        &\geq(\gamma(P_w)\cdot y_{w,j-1}-\gamma(P_w)c_w(C))\frac{1}{\gamma(P_wC)}\\
        &=\frac{y_{w,j-1}-c_w(C)}{\gamma(C)}\\
        &>\frac{y_{w,j-1}-y_{w,j-1}(1-\gamma(C))}{\gamma(C)}\\
        &=y_{w,j-1}.
    \end{align*}
    Since $y_{w,j-1}$ is at least $\max_{Q\in\paths_{v,w}^{j-1}} \left\{\frac{\xi_v-c(Q)}{\gamma(Q)}\right\}$,
    we conclude that $P_wC\notin \paths_{v,w}^{j-1}$, i.e., the walk $P_wC$ has at least~$j$ edges.
    Since $P_w$ is a simple path, and $C$ is a simple cycle, 
    $P_wC$ has $|V(P_w)|+|V(C)|$ edges, so we obtain $2|V(P_w\cup C)|\geq |V(P_w)|+|V(C)|\geq j$.
    As a result $|V(P_w\cup C)|\geq j/2$.
  \end{proof}

  \begin{lemma}\label{l:rw3}
    If for some $v$ we have $\xmax{v}<\xi_v$, the algorithm will terminate with a certificate.
  \end{lemma}
  \begin{proof}
    The ``$\implies$'' part is obvious. For the ``$\impliedby$'' part we refer to~\cite[Lemma~4.6]{RestrepoW09}.
    We only argue that our problem is equivalent to the negative-cost generalized augmenting
    path detection problem, as defined in~\cite{RestrepoW09}. In this problem, one
    is given a target vertex $t\in V$ and is asked to find a so-called \emph{negative GAP},
    which is a closed walk $C=w\to w$ with $\gamma(C)>1$ along with
    a walk $P=w\to t$ such that $c(P)<\cb_w(C)$.

    Consider the graph $G'=(V',E')$, where $V'=V'\cup \{t\}$, with functions $c':E'\to \mathbb{R}$, $\gamma':E'\to\mathbb{R}$.
    For each $uv=e\in E$, we have a corresponding $uv=e'\in E'$ with
    $c'(e)=c(e)/\gamma(e)$ and $\gamma'(e)=1/\gamma(e)$.
    Moreover, for each $v\in V$, we have an edge $vt=e'_v\in E'$
    with $c'(e'_v)=-\xi_v$, $\gamma'(e'_v)=1$.
    Note that we have a negative GAP in $G'$ iff for some walk $P=w\to v$
    and a cycle $C$ with $\gamma'(C)>1$, we have
    $c'(Pe'_v)<\cb'_w(C)$, or equivalently $c'(P)-\gamma'(P)\cdot \xi_v<\frac{c'_w(C)}{1-\gamma'(C)}$.
    This in turn is equivalent to
    \begin{equation*}
      \frac{1}{\gamma(P^R)}\cdot (c(P^R)-\xi_v)<\frac{c_w(C^R)}{\gamma(C^R)-1}.
    \end{equation*}
    Therefore we have, $c(P^R)+\gamma(P^R)\cdot \cb_w(C^R)<\xi_v$ and $P^R$ is a $v\to w$ walk, and $C^R$ is a $w\to w$ closed walk
    with $\gamma(C^R)<1$. This proves that there is a 1-1 correspondence
    between negative GAPs in $G'$, and certificates that we seek for in $G$.
    
    Finally, the algorithm that we described is equivalent to that of~\cite[Figure~2]{RestrepoW09}:
    applied to $G'$, that algorithm computes walks $Q_{w,j}'=w\to v$ minimizing $c'(Q_{w,j}')-\gamma'(Q_{w,j}')\cdot \xi_v$,
    which is equivalent to maximizing $\xi_v-c'(Q_{w,j}')=\frac{1}{\gamma((Q_{w,j}')^R)}\left(\xi_v-c((Q_{w,j}')^R)\right)$.
    So by substituting $Q_{w,j}:=(Q_{w,j}')^R$ the algorithm computes walks $Q_{w,j}=v\to w$ in $G$ 
    maximizing $\frac{1}{\gamma(Q_{w,j}}\left(\xi_v-c(Q_{w,j})\right)$, which is precisely what
    our version of the algorithm does. The tie-braking rules are the same as well.
  \end{proof}

  To guarantee consistency, when picking a vertex $w$ in the above lemma to produce
  a certificate $(P_w,C)$, we use $w$ with $y_{w,j-1}>\cb_w(C)$
  that is minimal according to some fixed order on $V$.

  \begin{lemma}
    Let $\xmax{v}<\xi_v$ and $\xmax{u}\geq \xi_u$ for all $u\in V\setminus\{v\}$.
    Then, the algorithm terminates with a certificate $(P_w,C)$ that does not depend
    on $\{\xi_u:u\in V\setminus\{v\}\}$.
  \end{lemma}
  \begin{proof}
    Suppose the algorithm run for $\xi$ produces a certificate $(P_w,C)$ in phase $j$.
    Recall that by construction, $P_w=e_1\ldots e_k$, where $e_i=v_{i-1}v_i$, and $v_{0}=v$, $v_k=w$,
    and $y_{v,j-1}=\xi_v$.
    For any $i$, we have 
    \begin{equation*}
      y_{v_i,j-1}=(y_{v_{i-1},j-2}-c(e_i))/\gamma(e_i)\leq (y_{v_{i-1},j-1}-c(e_i))/\gamma(e_i).
    \end{equation*}
    Therefore, by chaining such inequalities, for any $i$ we have:
    \begin{equation*}
      y_{w,j-1}\leq (y_{v_{i-1},j-1}-c(e_i\ldots e_k))/\gamma(e_i\ldots e_k).
    \end{equation*}

    Suppose for some $v_i$, $i\geq 0$, $Q_{v_i,j-1}$ is not a $v\to v_i$ path, but a $u\to v_i$
    path, for some $u\neq v$.
    Then,
    \begin{equation*}
      y_{v_i,j-1}=(\xi_u-c(Q_{v_i,j-1}))/\gamma(Q_{v_i,j-1}),
    \end{equation*}
    and thus
    \begin{equation*}
      y_{w,j-1}\leq (\xi_u-c(Q_{v_{i},j-1}e_{i+1}\ldots e_k))/\gamma(Q_{v_{i},j-1}e_{i+1}\ldots e_k).
    \end{equation*}
    As a result, we obtain:
    \begin{equation*}
      c(Q_{v_{i},j-1}e_{i+1}\ldots e_k)+\gamma(Q_{v_{i},j-1}e_{i+1}\ldots e_k)\cdot y_{w,j-1}\leq \xi_u.
    \end{equation*}
    This, in combination with $y_{w,j-1}>\cb_w(C)$, shows that $\xmax{u}<\xi_u$, a contradiction.
    This proves that $Q_{v_i,j-1}$ is a $v\to v_i$ path for all $i$.

    Next we prove that for each $z\in V(C)$, $Q_{z,j}$ is a $v\to z$ walk.
    Let $C=e_0\ldots e_{k-1}$ such that $e_i=u_iu_{(i+1)\bmod{k}}$.
    Split the sequence of vertices appearing on $C$ (in order) into maximal segments $u_g,\ldots,u_h$ such that
    $y_{u_g,j}>y_{u_{g},j-1}$ and $y_{u_l,j}=y_{u_l,j-1}$ for all
    $l=g+1,\ldots,h$.
    Such a split is possible since $y_{w,j}>y_{w,j-1}$.
    For each segment, we have $y_{u_{(h+1)\bmod k},j}>y_{u_{(h+1)\bmod k},j-1}$.
    
    Consider an arbitrary such segment. We have already proved
    that $Q_{u_{(h+1)\bmod k},j-1}$ is a $v\to u_{(h+1)\bmod k}$ path, since the choice
    of the vertex $w$ only required $y_{w,j}>y_{w,j-1}$.
    We also have $Q_{u_h,j}=Q_{u_{(h+1)\bmod k},j-1}\cdot e_h$ since $p_{u_h,j}=e_h$.
    Hence, $Q_{u_h,j}$ is a $v\to u_h$ path.
    But $y_{u_h,j}=y_{u_h,j-1}$, so in fact $Q_{u_h,j-1}=Q_{u_h,j}$ is also a $v\to u_h$ path.
    Analogously, we inductively prove that subsequently $Q_{u_{h-1},j}=Q_{u_{h-1},j-1}, Q_{u_{h-2},j}=Q_{u_{h-2,j-1}}, \ldots, Q_{u_{g+1},j}=Q_{u_{g+1},j-1}$
    are also walks starting in $v$.
    Finally, from $Q_{u_{g+1},j-1}=v\to u_{g+1}$ we conclude that $Q_{u_g,j}=v\to u_g$.

    Now, consider a vector $(\xi^*)_{u\in V}$ such that $\xi_v^*=\xi_v$
    and $\xi^*_u=\xmax{u}$ for all $u\in V\setminus\{v\}$.
    For any other vector $(\xi')_{v\in V}$ that satisfies
    the requirements of the lemma, we have $\xi'_u\leq \xi^*_u$ for all $u\in V$.
    By Lemma~\ref{l:rw3}, running the algorithm with the input $\xi^*$
    produces some certificate $(P_s^*,C^*)$ for some $s\in V$, and
    running it with the input $\xi'$ produces a certificate $(P_t',C')$ for some
    $t\in V$.

    Denote by $Q_{w,j}^*,y_{w,j}^*$ and $Q_{w,j}',y_{w,j}'$ the respective paths $Q_{w,j}$ and values $y_{w,j}$
    for the inputs $\xi^*$ and $\xi'$ respectively.
    Note that for any $w,j$ we have $y_{w,j}^*\geq y_{w,j}'$ by the definition of $\xi^*$.
    Moreover, since $\xi^*_v=\xi'_v$, we have that $Q_{w,j}^*=v\to w$ implies
    that $Q_{w,j}'=Q_{w,j}^*$.
    
    Suppose that the certificate $(P_s^*,C^*)$ is produced in phase $j$ and $s$ is chosen minimal possible.
    Then, as we have proved, for all $z\in V(P_s^*)$, $Q_{z,j-1}^*$ is a $v\to z$ walk.
    Moreover, for all $z\in V(C^*)$, $Q_{z,j}^*$ is a $v\to z$ walk.
    Consequently, $Q_{z,j-1}'=Q_{z,j-1}^*$ for all $z\in V(P_s^*)$, and $Q_{z,j}'=Q_{z,j}^*$
    for all $z\in V(C^*)$.
    As a result, the same cycle $C^*$ exists in $H_j$ and the same path $P_s^*$
    exists in $H_{j-1}$ for the input $\xi'$. We conclude that the algorithm
    run for $\xi'$ will terminate with $(P_s^*,C^*)$ unless it terminates earlier,
    i.e., in an earlier phase or in the phase $j$ with a smaller vertex $t$.

    Now suppose that for the input $\xi'$, the algorithm terminates earlier (at phase $j$)
    than if it was run with the input $\xi^*$.
    Note that then we need to either have (i) $Q_{z,j-1}^*=w\to z$ for some $z\in V(P_t')$ and $w\in V\setminus\{v\}$,
    or we have (ii) $Q_{z,j}^*=w\to z$ for some $z\in V(C')$ and $w\in V\setminus\{v\}$ --
    otherwise the run with the input $\xi^*$ would have terminated no later than that for $\xi'$.
    
    Consider the case (i). Then, $y_{z,j-1}^*\geq y_{z,j-1}'$ and if $P_t'=R_1R_2$, where $R_1=v\to z$ and $R_2=z\to t$,
    we would obtain similarly as before $c(Q_{z,j-1}\cdot R_2)+\gamma(Q_{z,j-1}\cdot R_2)\cdot y'_{t,j-1}\leq \xi_w$,
    which, along with $y'_{t,j-1}<\cb_t(C')$ would imply $\xmax{w}<\xi_w$, a contradiction.

    Now let us focus on (ii). Let $C'=C_1C_2$, where $t\to z=C_1=e_0\ldots e_l$ and $C_2=z\to t$.
    Let $e_i=u_iu_{i+1}$, so that $u_0=t$ and $u_{l+1}=z$.
    By chaining the inequalities on $C_2$ we obtain
    $\cb_t(C')<y'_{t,j-1}<y'_{t,j}\leq (y'_{z,j}-c(C_2))/\gamma(C_2)\leq (y^*_{z,j}-c(C_2))/\gamma(C_2)=(\xi^*_w-c(Q_{z,j}^*C_2))/\gamma(Q_{z,j}^*C_2)$.
    As a result, we have $c(Q_{z,j}^*C_2)+\gamma(Q_{z,j}^*C_2)\cdot \cb_t(C')<\xi^*_w$.
    So, $\xmax{w}<\xi^*_w$, a contradiction.

    We conclude that for any input $\xi'$, the algorithm has to terminate with the same certificate
    as for the input $\xi^*$.
    As a result, it terminates with the same certificate for any $\xi$ satisfying the requirements
    of the lemma.
  \end{proof}

\paragraph{Acknowledgments.}
We thank the anonymous reviewers for their very useful comments.

\small
\bibliographystyle{alpha}
\bibliography{references}

\end{document}